\documentclass[sn-mathphys]{sn-jnl_Arxiv}

\usepackage{comment}
\usepackage{mathtools}
\normalbaroutside
\usepackage{glossaries}
\usepackage{caption}
\usepackage{subcaption}
\usepackage{natbib}

\newcommand{\EE}{\mathbb{E}}
\newcommand{\RR}{\mathbb{R}}
\newcommand{\One}{\mathbf{1}}
\newcommand{\intd}{\, \text{d}}
\newcommand{\vm}[1]{\ensuremath{\mathbf{#1}}}
\newcommand{\vms}[1]{\boldsymbol{#1}}
\newcommand{\transpose}{\top}

\newcommand{\argmax}{\ensuremath{\text{argmax}}}

\newcommand{\pdf}{\text{pdf}}
\newcommand{\Nmc}{N_{\text{MC}}} 
\newcommand{\SD}{\Omega} 

\newcommand{\lb}{\text{\textbf{lb}}}
\newcommand{\ub}{\text{\textbf{ub}}}

\newcommand{\train}{\mathcal{T}}
\newcommand{\trainH}{\mathcal{T}_{\text{\normalfont{H}}}}
\newcommand{\trainI}{\mathcal{T}_{\text{\normalfont{I}}}}
\newcommand{\nd}{n_{\text{{\fontsize{2.5}{4}\selectfont \normalfont{kd}}}}} 
\newcommand{\kd}{n_{\text{{\fontsize{2.5}{4}\selectfont \normalfont{kd}}}}} 
\newcommand{\kds}{n_{\text{{\fontsize{2.5}{4}\selectfont \normalfont{2kd}}}}} 

\newcommand{\xopt}{\vm x^{\text{opt}}}
\newcommand{\Msr}{\vm M_{\text{H}}}
\newcommand{\bsr}{\vm b_{\text{H}}}
\newcommand{\vsr}{\vm v^{(k)}}
\newcommand{\MH}{\vm M_{\text{\normalfont{H}}}}
\newcommand{\bH}{\vm b_{\text{H}}}

\newcommand{\MHB}{\vm M_{\text{\normalfont{HB}}}}
\newcommand{\bHB}{\vm b_{\text{\normalfont{HB}}}}

\newcommand{\MB}{\vm M_{\text{\normalfont{B}}}}
\newcommand{\bB}{\vm b_{\text{\normalfont{B}}}}

\newcommand{\reffull}{\text{ref}_{\text{full}}}
\newcommand{\refB}{\text{ref}_{\text{B}}}
\newcommand{\refN}{\text{ref}_{\text{N}}}

\DeclarePairedDelimiter{\abs}{\lvert}{\rvert}

\newacronym{DFO}{DFO}{derivative-free optimization}
\newacronym{SQP}{SQP}{sequential quadratic programming}
\newacronym{MC}{MC}{Monte Carlo}

\jyear{2022}%

\theoremstyle{thmstyleone}%
\newtheorem{theorem}{Theorem}
\newtheorem{lemma}[theorem]{Lemma}%

\theoremstyle{thmstyletwo}%
\newtheorem{remark}{Remark}%

\theoremstyle{thmstylethree}%
\newtheorem{definition}{Definition}%

\begin{document}

\title[Hermite-type modifications of BOBYQA]{Hermite-type modifications of BOBYQA for optimization with some partial derivatives}


\author*[1,2]{\fnm{Mona} \sur{Fuhrländer}}\email{mona.fuhrlaender@tu-darmstadt.de}

\author[1,2]{\fnm{Sebastian} \sur{Schöps}}\email{sebastian.schoeps@tu-darmstadt.de}


\affil[1]{\orgdiv{Computational Electromagnetics Group (CEM)}, \orgname{Technische Universität Darmstadt}, \orgaddress{\street{Schlossgartenstr. 8}, \city{Darmstadt}, \postcode{64289}, \state{Germany}}}

\affil[2]{\orgdiv{Centre for Computational Engineering (CCE)}, \orgname{Technische Universität Darmstadt}, \orgaddress{\street{Dolivostr. 15}, \city{Darmstadt}, \postcode{64293}, \state{Germany}}}



\abstract{In this work we propose two Hermite-type optimization methods, Hermite least squares and Hermite BOBYQA, specialized for the case that some partial derivatives of the objective function are available and others are not. The main objective is to reduce the number of objective function calls by maintaining the convergence properties. Both methods are modifications of Powell's derivative-free BOBYQA algorithm. But instead of (underdetermined) interpolation for building the quadratic subproblem in each iteration, the training data is enriched with first and -- if possible -- second order derivatives and then (weighted) least squares regression is used. Proofs for global convergence are discussed and numerical results are presented. Further, the applicability is verified for a realistic test case in the context of yield optimization.
Numerical tests show that the Hermite least squares approach outperforms classic BOBYQA if half or more partial derivatives are available. In addition, the Hermite-type approaches achieve more robustness and thus better performance in case of noisy objective functions.}




\maketitle

\section{Introduction}
\label{sec:intro}

In optimization we typically distinguish between gradient based and \gls{DFO} approaches. If gradients are available, they provide helpful information about the descent in specific points and thus, can improve the performance of the algorithm. With performance of the optimization algorithm we refer to the ability of reaching the optimal solution and the computational effort in order to reach this solution, which can be measured by the number of objective function calls. However, often gradients are not available. In this research we are interested in optimization problems where the objective function is expensive to evaluate, e.g., finite element simulations. Thus, we focus on the possibility of reducing the number of objective function calls by using all information we have.

In case of multidimensional optimization it often occurs that for some directions the partial derivatives are available and for others not. 
In sec.~\ref{sec:PracticalExample} we will provide a practical example in the context of yield optimization.
In this work we provide an optimization strategy well suited to benefit from the known derivatives, without requiring derivatives (or approximations of derivatives) for each direction.

Commonly used methods for solving nonlinear optimization problems are based on 
approximating the objective function quadratically in each iteration and solving this subproblem in a trust region in order to obtain the new iterate solution for the original problem.
In \gls{SQP} gradients are used to build the quadratic approximation based on second order Taylor expansions~\cite[Chap. 19]{Ulbrich_2012aa}. For \gls{DFO} Powell proposed the BOBYQA (Bound constrained Optimization BY Quadratic Approximation) method using polynomial interpolation for building the quadratic approximation~\cite{Powell_2009aa}. For both, \gls{SQP} and BOBYQA, there are several modifications and various popular implementations, see e.g.~\cite{Cartis_2019aa,Cartis_2021aa,Powell_COBYLA,Powell_LINCOA,Kraft_1988aa}.
Other \gls{DFO} methods based on quadratic approximations and additionally considering noisy objective data are proposed for example in~\cite{Cartis_2019aa} using sample averaging and restarts, in~\cite{Billups_2013} using weighted regression, and in~\cite{Menhorn_2022aa}, using Gaussian process regression. Besides methods based on quadratic approximations, particle swarm~\cite{Kennedy_1995} and genetic algorithms~\cite{Audet_2017aa} are frequently used for \gls{DFO}.

However, to our best knowledge, all these methods have in common that they use derivatives for \textit{all} directions or for \textit{none}.
If some derivatives are available, \gls{SQP} would approximate the missing ones using for example finite differences, classic \gls{DFO} methods would ignore all derivatives. 
Finite differences require at least one additional function evaluation per direction each time the gradient has to be calculated. Especially for higher dimensional problems, this leads to an enormous increase of the computational costs. Further, finite differences approximations are sensitive to noisy data. On the other hand, we assume that BOBYQA could perform better, i.e., would need less iterations and thus function evaluations, if we would provide all information we have. For that reason we propose to modify BOBYQA to enable the usage of \textit{some} derivative information. More precisely, we extend the Python implementation PyBOBYQA by Cartis et al.~\cite{Cartis_2019aa}. We propose two variants and, in accordance with Hermite interpolation, cf.~\cite[Chap. 6.6]{Hermann_HermiteInterpol} or~\cite{Sauer_1995aa}, we call them Hermite least squares and Hermite BOBYQA.
Further we investigate the impact of noisy objective functions and observe higher robustness compared to the original BOBYQA and \gls{SQP}.

This work is structured as follows. We start with the formulation of the problem setting in sec.~\ref{sec:setting} and an introduction into~\gls{DFO} and BOBYQA in sec.~\ref{sec:DFO}. In sec.~\ref{sec:firststep} we propose the Hermite least squares and Hermite BOBYQA approaches and provide numerical results in sec.~\ref{sec:NumericalTests}. We conclude the paper with a practical example from the field of electrical engineering in sec.~\ref{sec:PracticalExample} and some final remarks.

\section{Problem setting}
\label{sec:setting}
Even though \gls{SQP} is able to handle general nonlinear constraints, BOBYQA, as indicated by its name, only accepts bound constraints. Powell also proposed two more methods, LINCOA (LINearly Constrained Optimization Algorithm), which allows linear constraints, and COBYLA (Constrained Optimization BY Linear Approximations), which allows general constraints but uses only linear approximations~\cite{Powell_LINCOA,Powell_COBYLA}.
However, we focus on BOBYQA and thus,
we consider a bound constrained optimization problem with multiple optimization variables, i.e., for a function $f:\RR^n \rightarrow \RR$, an optimization variable $\vm x \in \RR^n$ and lower and upper bounds $\vm x_{\text{lb}}, \vm x_{\text{ub}} \in \RR^n$ the optimization problem reads
\begin{align}
&\min_{\vm x \in \RR^n} f(\vm x)\\
\label{eq:OptProb}
&\text{ s.t. } \vm x_{\lb} \leq \vm x \leq \vm x_{\ub}. \notag
\end{align}
We assume that the derivatives with respect to some directions $x_i$, $i\in\mathcal{I}=\lbrace 1,\dots,n\rbrace$ are known, others are not. 
We denote the index set of known first order derivative directions by $\mathcal{I}_{\text{d}} \subseteq \mathcal{I}$, such that the set of available first order derivatives is defined by
\begin{equation}
\mathcal{D} := \left\lbrace\frac{\partial}{\partial x_i}f\right\rbrace_{i\in \mathcal{I}_{\text{d}}}.
\label{eq:Setting_SetDerivs}
\end{equation}
In order to define the set of known second order derivatives, we introduce the 
tuple set $\mathcal{I}_{\text{2d}}\subseteq \mathcal{I} \times \mathcal{I}$.
Then, the set of available second order derivatives is given by
\begin{equation}
\mathcal{D}_2 := \left\lbrace\frac{\partial^2}{\partial x_i \partial x_j}f\right\rbrace_{(i,j) \in \mathcal{I}_{\text{2d}}}.
\label{eq:Setting_SetDerivs2nd}
\end{equation}
For the sake of simplicity we will focus on the practically relevant case of first order derivatives. However, the proposed methods can be straightforwardly adjusted for using the second order derivatives, cf. Sec.~\ref{sec:Num2ndderivs} and Appendix~\ref{sec:Hls_2ndDerivs}. Since we build a quadratic approximation, higher order derivatives are not of concern.
In the remainder of this paper, for better readability and without limitation of generality we assume that the $x_i$ are ordered such that we can define
\begin{equation}
\mathcal{I}_{\text{d}} = \lbrace 1,\dots,\nd\rbrace, \ \nd\leq n
\label{eq:IndexsetKnownDerivs}
\end{equation}
as the index set of directions for which we consider the first partial derivative to be known.

\section{Derivative-free optimization}
\label{sec:DFO}
Powell's BOBYQA algorithm is a widely used algorithm in the field of \gls{DFO}~\cite{Powell_2009aa}. The original implementation is in Fortran. Cartis et al. published a Python implementation called PyBOBYQA~\cite{Cartis_2019aa,Cartis_2021aa}. It contains some simplifications and several modifications (e.g. for noisy data and global optimization), but Powell's main ideas remain unchanged. In this work, on the programming side, we use PyBOBYQA as a basis and add some new features to it. While the original BOBYQA method (and also PyBOBYQA) are efficient in practice, 
it cannot be proven that they converge globally, i.e., that they converge from an arbitrary starting point to a stationary point~\cite[Chap. 10.3]{Conn_2009book}.
Conn et al. reformulated the BOBYQA method in~\cite[Chap. 11.3]{Conn_2009book}, maintaining the main concept, but enabling a proof of convergence -- on the cost of practical efficiency. For the theoretical considerations in this work, we take Conn's reformulation as a basis.
Before we come to our modifications for mixed gradient information, we recall the basics of \gls{DFO} methods and BOBYQA.

\subsection{Notation}
\label{sec:Notation}
Let $m(\vm x)$ be a polynomial of degree $d$ with $\vm x \in \RR^n$ and let $\Phi=\lbrace \phi_0(\vm x), \dots, \phi_q(\vm x) \rbrace$ be a basis in $\mathcal{P}_n^d$ and $q_1=q+1$. Further, we define the vector of basis polynomials evaluated in $\vm x$ by $\vms \Phi(\vm x):=(\phi_0(\vm x), \dots, \phi_q(\vm x))^{\transpose}$. The training data set is denoted by
\begin{equation}
\train = \lbrace(\vm y^0,f(\vm y^0)),\dots, (\vm y^p,f(\vm y^p))
\label{eq:Notation_TrainingDataSet}
\end{equation}
and $p_1 = p+1$. 
The original BOBYQA method is based on interpolation, however, we formulate the problem more generally as interpolation or least squares regression problem.
The system matrix and the right hand side of the interpolation or regression problem are then given by
\begin{equation}
\vm M \equiv \vm M(\Phi,\train) = 
\left(\begin{matrix}
\phi_0(\vm y^0) & \dots & \phi_q(\vm y^0)\\
\vdots & & \vdots \\
\phi_0(\vm y^p) & \dots & \phi_q(\vm y^p)
\end{matrix}\right)
\ \text{ and } \
\vm b \equiv \vm b(\train) = 
\left(\begin{matrix}
f(\vm y^0)\\
\vdots \\
f(\vm y^p)
\end{matrix}\right).
\label{eq:SystemMatrixRHSGeneral}
\end{equation}
If $p_1 = q_1$ the system matrix $\vm M$ is quadratic and $\vm v \in \RR^{q_1=p_1}$ solves the interpolation problem
\begin{equation}
\vm M \vm v = \vm b.
\label{eq:InterpolProblemGeneral}
\end{equation}
If $p_1 > q_1$ the system matrix $\vm M$ is in $\RR^{p_1\times q_1}$ and $\vm v \in \RR^{q_1}$, this leads to an overdetermined interpolation problem and can be solved with least squares regression 
\begin{equation}
\vm M \vm v \stackrel{\text{l.s}}{=} \vm b 
\quad \Leftrightarrow \quad
\min_{\vm v \in \RR^{q_1}} ||\vm M \vm v - \vm b||^2
\quad \Leftrightarrow \quad
\vm M^{\transpose}\vm M \vm v = \vm M^{\transpose} \vm b.
\label{eq:RegressProblemGeneral}
\end{equation}
If the matrix $\vm M$ in~\eqref{eq:InterpolProblemGeneral} is non-singular, the linear system~\eqref{eq:InterpolProblemGeneral} has a unique solution. Then, following~\cite[Chap. 3]{Conn_2009book}, the  corresponding training data set is said to be \textit{poised}. Analogously, if the matrix $\vm M$ in~\eqref{eq:RegressProblemGeneral} has full column rank, the linear system~\eqref{eq:RegressProblemGeneral} has a unique solution. And, following~\cite[Chap. 4]{Conn_2009book}, the  corresponding training data set is said to be \textit{poised for polynomial least-squares regression}. When talking about training data sets in the following, we always assume them to be poised, if not specifically noted otherwise.

Although many of the results hold for any choice of a basis, in the following we use the monomial basis of degree $d=2$. Thus, if not mentioned differently, for the remainder of this paper, $\Phi$ is defined by the $(n+1)(n+2)/2$-dimensional basis
\begin{equation}
\Phi = \left\lbrace 1, x_1, \dots, x_n, \frac{1}{2}x_1^2, x_1x_2, x_1x_3, \dots, x_{n-1}x_n, \frac{1}{2}x_n^2 \right\rbrace.
\label{eq:MonomialBasis2}
\end{equation}

\subsection{$\Lambda$-poisedness}
\label{sec:LambdaPoisedDef}
Many \gls{DFO} algorithms are model based. Thus, in order to guarantee global convergence, we have to ensure that the model is \textit{good enough}. In gradient based methods, typically some Taylor expansion error bounds are considered. In \gls{DFO} methods, which are based on interpolation or regression, the quality of the model depends on the quality of the training data set. This leads us to the introduction of the term \textit{$\Lambda$-poisedness}. We recall the definitions of $\Lambda$-poisedness from~\cite[Def. 3.6, 4.7, 5.3]{Conn_2009book}.
\begin{definition}[$\Lambda$-poisedness in the interpolation sense]\hfill \\
	\label{def:LambdaPoisedInterpol}
	Given a poised interpolation problem as defined in sec.~\ref{sec:Notation}. Let $\mathcal{B}\subset \RR^n$ and $\Lambda>0$. Then the training data set $\train$ is $\Lambda$-poised in $\mathcal{B}$ (in the interpolation sense) if and only if
	\begin{equation}
	\forall \vm x \in \mathcal{B} \ \exists \vm{l}(\vm x) \in \RR^{p_1} \text{ s.t. } \ \
	\sum_{i=0}^{p} l^i(\vm x) \vms \Phi(\vm y^i) = \vms \Phi(\vm x) \ \ \text{ with } \ \ ||\vm l(\vm x)||_{\infty} \leq \Lambda. 
	\label{eq:LambdaPoisedInterpol}
	\end{equation}
\end{definition}
\begin{remark}\label{remark:LagrangePolynomialsInterpol}
	Note that in case of using the monomial basis, the equality in Def.~\ref{def:LambdaPoisedInterpol} can be rewritten as $\vm M^{\transpose}\vm l(\vm x) = \vms \Phi(\vm x)$ and the $l^i(\vm x)$, $i=0,\dots,p$ are uniquely defined by the Lagrange polynomials and can be obtained by solving
	\begin{equation}
	\vm M \vms \lambda^i = \vm e^{i+1},  
	\label{eq:LagrangePolynomialInterpol}
	\end{equation}
	where $\vm e^i\in \RR^{p_1}$ denotes the $i$-th unit vector and the elements of $\vms \lambda_{i}$ are the coefficients of the polynomial $l^i$, which will be evaluated at $\vm x$.
\end{remark}
\begin{definition}[$\Lambda$-poisedness in the regression sense]\hfill \\
	\label{def:LambdaPoisedRegress}
	Given a poised regression problem as defined in sec.~\ref{sec:Notation}. Let $\mathcal{B}\subset \RR^n$ and $\Lambda>0$. Then the training data set $\train$ is $\Lambda$-poised in $\mathcal{B}$ (in the regression sense) if and only if
	\begin{equation}
	\forall \vm x \in \mathcal{B} \ \exists \vm{l}(\vm x) \in \RR^{p_1} \text{ s.t. } \ \
	\sum_{i=0}^{p} l^i(\vm x) \vms \Phi(\vm y^i) = \vms \Phi(\vm x) \ \text{ with } \ ||\vm l(\vm x)||_{\infty} \leq \Lambda. 
	\label{eq:LambdaPoisedRegress}
	\end{equation}
\end{definition}
\begin{remark}\label{remark:LagrangePolynomialsRegress}
	Note that the $l^i(\vm x)$, $i=0,\dots,p$ are not uniquely defined since the system in~\eqref{eq:LambdaPoisedRegress} is underdetermined. However, the minimum norm solution corresponds to the Lagrange polynomials (in the regression sense), cf.~\cite[Def. 4.4]{Conn_2009book}. Analogously to remark~\ref{remark:LagrangePolynomialsInterpol} they can be computed by solving
	\begin{equation}
	\vm M \vms \lambda^i \stackrel{\text{l.s.}}{=} \vm e^{i+1} 
	\label{eq:LagrangePolynomialRegress}
	\end{equation}
	and using the entries of $\vms \lambda^i$ as coefficients for the polynomial $l^i$.
\end{remark}

In the next section we will see that BOBYQA is based on the minimum norm solution of underdetermined interpolation problems. Thus, we recall the definition of $\Lambda$-poisedness in the minimum-norm sense.
\begin{definition}[$\Lambda$-poisedness in the minimum-norm sense]\hfill \\
	\label{def:LambdaPoisedMinNorm}
	Given a poised underdetermined interpolation problem as defined in sec.~\ref{sec:Notation}, but with $p<q$. Let $\mathcal{B}\subset \RR^n$, $\Lambda>0$ and $\Phi$ the monomial basis. Then the training data set $\train$ is $\Lambda$-poised in $\mathcal{B}$ (in the minimum-norm sense) if and only if
	\begin{equation}
	\forall \vm x \in \mathcal{B} \ \exists \vm{l}(\vm x) \in \RR^{p_1} \text{ s.t. } \ \
	\sum_{i=0}^{p} l^i(\vm x) \vms \Phi(\vm y^i) \stackrel{\text{\normalfont{l.s.}}}{=} \vms \Phi(\vm x) \ \text{ with } \ ||\vm l(\vm x)||_{\infty} \leq \Lambda. 
	\label{eq:LambdaPoisedMinNorm}
	\end{equation}
\end{definition}
\begin{remark}
	Note that the $l^i(\vm x)$, $i=0,\dots,p$ exist and are uniquely defined if the system matrix $\vm M$ has full column rank and correspond to the Lagrange polynomials (in the minimum norm sense), cf.~\cite[Def. 5.1]{Conn_2009book}. Analogously to remark~\ref{remark:LagrangePolynomialsInterpol} and~\ref{remark:LagrangePolynomialsRegress} they can be computed by the minimum norm solution of
	\begin{equation}
	\vm M \vms \lambda^i \stackrel{\text{m.n.}}{=} \vm e^{i+1} 
	\label{eq:LagrangePolynomialMinNorm}
	\end{equation}
	and using the entries of $\vms \lambda^i$ as coefficients for the polynomial $l^i$.
\end{remark}
It can be shown, that the poisedness constant $\Lambda$, or rather $1/\Lambda$ can be interpreted as the distance to singularity of the matrix $\vm M$ or to linear dependency of the vectors $\vms \Phi(\vm y^i)$, $i=0,\dots,p$~\cite[Chap. 3]{Conn_2009book}.

\subsection{BOBYQA}
\label{sec:BOBYQA}
The following description of BOBYQA's main ideas follows the one in~\cite{Cartis_2021aa}.
The BOBYQA algorithm is a trust region method based on a (typically underdetermined) quadratic interpolation model. The training data set defined in~\eqref{eq:Notation_TrainingDataSet} contains the objective function evaluations for each sample point $\vm y^i$, $i=0,\dots,p$, and $\abs{\train}\in\left[n+2,(n+1)(n+2)/2\right]$. Note in PyBOBYQA setting $\abs{\train}=n+1$ is possible, but then a fully linear interpolation model is applied.

Let $\vm x^{(k)}$ denote the solution at the $k$-th iteration. At each iteration a local quadratic model is built, i.e.,
\begin{equation}
f(\vm x) \approx m^{(k)}(\vm x) = c^{(k)} + {\vm g^{(k)}}^{\transpose}(\vm x - \vm x^{(k)}) +\frac{1}{2}(\vm x - \vm x^{(k)})^{\transpose}\vm H^{(k)} (\vm x - \vm x^{(k)}),
\label{eq:BOBYQA_QP}
\end{equation}
fulfilling the interpolation conditions
\begin{equation}
f(\vm y^j) = m^{(k)}(\vm y^j) \ \ \forall \vm y^j \in \train.
\label{eq:BOBYQA_InterpolCond}
\end{equation}
For $|\train|=(n+1)(n+2)/2$ the interpolation problem is fully determined. For $|\train|<(n+1)(n+2)/2$ the remaining degrees of freedom are set by minimizing the distance between the current and the last approximation of the Hessian $\vm H$ in the matrix Frobenius norm, i.e., 
\begin{equation}
\min_{c^{(k)},\vm g^{(k)}, \vm H^{(k)}} ||\vm H^{(k)}-\vm H^{(k-1)}||_{\text{F}}^2 \ \ \text{ s.t. \eqref{eq:BOBYQA_InterpolCond} holds},
\label{eq:BOBYQA_MinFrobNorm}
\end{equation}
where typically $\vm H^{(-1)} = \vm 0^{n\times n}$.
Once the quadratic model is built, the trust region subproblem 
\begin{align}
&\min_{\vm x \in \RR^n} m^{(k)}(\vm x)  \label{eq:BOBYQA_TRsubproblem}\\
&\text{ s.t. } ||\vm x - \vm x^{(k)}||_2\leq \Delta^{(k)} \notag
\end{align}
is solved, where $\Delta^{(k)}> 0 $ denotes the trust region radius. Then, having the optimal solution $\xopt$ of~\eqref{eq:BOBYQA_TRsubproblem} in the $k$-th iteration calculated, we check if the decrease in the objective function is sufficient. Therefore, the ratio between actual decrease and expected decrease
\begin{equation}
r^{(k)}=\frac{\text{actual decrease}}{\text{expected decrease}} = \frac{f(\vm x^{(k)})- f(\xopt)}{m^{(k)}(\vm x^{(k)})-m^{(k)}(\xopt)}
\label{eq:BOBYQA_DecreaseRatio}
\end{equation}
is calculated. If the ratio $r^{(k)}$ is sufficiently large, the step is accepted ($\vm x^{(k+1)}=\xopt$) and the trust region radius increased ($\Delta^{(k+1)}>\Delta^{(k)}$). Otherwise, the step is rejected ($\vm x^{(k+1)}=\vm x^{(k)}$) and the trust region radius decreased ($\Delta^{(k+1)}<\Delta^{(k)}$).

An important question is now, how to maintain the training data set. An accepted solution is added to $\train$,
i.e., $\train = \train \cup \lbrace (\vm y^{\text{add}},f(\vm y^{\text{add}})) \rbrace$ with $\vm y^{\text{add}} = \xopt$, 
and since $|\train|$ is fixed, another data point has to leave the training data set. Hereby, the aim is to achieve the best possible quality of the model. We know from sec.~\ref{sec:LambdaPoisedDef} that the quality of the model depends on the training data set and can be expressed by the poisedness constant $\Lambda$. Thus, the decision which point is replaced depends on its impact on the $\Lambda$-poisedness. 
Let 
$l^i$, $i=0,\dots,p$
be the Lagrange polynomials obtained by evaluating~\eqref{eq:LagrangePolynomialInterpol} or~\eqref{eq:LagrangePolynomialMinNorm} and 
%
%
\begin{equation}
i^{\text{go}}=\argmax_{i=0,\dots,p} l^i(\vm y^{\text{add}}).
\label{eq:BOBYQA_LeavingPoint}
\end{equation}
Then, the point $\vm y^{i^{\text{go}}}$ is replaced by the new iterate $\vm y^{\text{add}}$. This means, that the point with the worst (largest) value of the corresponding Lagrange polynomial, evaluated at the new iterate solution, is going to be replaced, i.e., the updated training data set is built by $\train = \train \cup \lbrace(\vm y^{\text{add}},f(\vm y^{\text{add}}))\rbrace \backslash \lbrace(\vm y^{i^{\text{go}}},f(\vm y^{i^{\text{go}}}))\rbrace$.

Regardless of the subproblem's optimal solution, sometimes points are exchanged, only in order to improve the training data set. Let $\vm y^i$ be the training data point which shall be replaced, e.g. the point furthest from the current optimal solution. Then we consider the corresponding Lagrange polynomial $l^i$ and choose a new point by solving
\begin{equation}
\vm y^{\text{new}}=\min_{\vm y \in \mathcal{B}} l^{i}(\vm y).
\label{eq:BOBYQA_NewPoint}
\end{equation}
For more details we refer to the original work by Powell~\cite{Powell_2009aa}.

\subsubsection{Solving the linear system}

Later we will modify the linear systems which are solved to determine $c^{(k)}$, $\vm g^{(k)}$ and $\vm H^{(k)}$ in~\eqref{eq:BOBYQA_QP}. To make it easier to follow, we recall the systems from~\cite{Cartis_2019aa}. We denote the current optimal solution by $\xopt\in \train$. 
Without limitation of generality we assume that $\xopt = \vm y^p$.
If we have a uniquely solvable interpolation problem (i.e. $p_1=q_1=(n+1)(n+2)/2$), the linear system to be solved reads as follows
\begin{equation}
\vm M_{\text{I}} \vm v^{(k)} = \vm b_{\text{I}}
\label{eq:BOBYQA_LGSfullydet}
\end{equation}
with
\begin{align}
&\vm M_{\text{I}} =
\left(\begin{matrix}
\phi_1(\vm y^0 - \xopt) & \dots & \phi_q(\vm y^0- \xopt)  \\
\vdots & & \vdots\\
\phi_1(\vm y^{p-1}- \xopt) & \dots & \phi_q(\vm y^{p-1}- \xopt) 
\end{matrix}\right) \in \RR^{p\times q}, 
\label{eq:BOBYQA_LGSfullydet_M}\\
&\vm v^{(k)} = 
\left(\begin{matrix}
\vm g^{(k)}\\
{\vm H^{(k)}}^{\star}
\end{matrix} \right) \in \RR^q \ \ \text{and} \ \ 
\vm b_{\text{full}} = 
\left(\begin{matrix}
f(\vm y^0) - f(\xopt) \\
\vdots\\
f(\vm y^{p-1}) - f(\xopt)
\end{matrix} \right) \in \RR^p,
\label{eq:BOBYQA_LGSfullydet_v_b}
\end{align}
where ${\vm H^{(k)}}^{\star}$ is a vector in $\RR^{(n^2+n)/2}$ containing the lower triangular and the diagonal elements of the diagonal matrix $\vm H^{(k)}$. The constant part is set to $c^{(k)} = f(\xopt)$.

In the case $n+2\leq p_1 <(n+1)(n+2)/2$ we need further information from~\eqref{eq:BOBYQA_MinFrobNorm}. The linear system we solve is formulated by
\begin{equation}
\MB \vm v^{(k)} = \vm \bB,
\label{eq:BOBYQA_LGSmfn}
\end{equation}
where
\begin{equation}
\MB
=\left(\begin{matrix}
\vm A & \vm B^{\transpose} \\
\vm B & \vm 0
\end{matrix}\right) \in \RR^{p+n\times p+n}, 
\label{eq:BOBYQA_LGSmfn_M}
\end{equation}
with $\vm A \in \RR^{p\times p}$ given by
\begin{equation}
\vm A
=\frac{1}{2} \left(\begin{matrix}
\left((\vm y^0-\xopt)^{\transpose}(\vm y^0-\xopt)\right)^2 &
\dots & \left((\vm y^0-\xopt)^{\transpose}(\vm y^{p-1}-\xopt)\right)^2 \\
\vdots & & \vdots \\
\left((\vm y^{p-1}-\xopt)^{\transpose}(\vm y^0-\xopt)\right)^2 &
\dots & \left((\vm y^{p-1}-\xopt)^{\transpose}(\vm y^{p-1}-\xopt)\right)^2 
\end{matrix}\right) 
\label{eq:BOBYQA_LGSmfn_Mquad}
\end{equation}
and
\begin{equation}
\vm B
=\left(\begin{matrix}
\vert & & \vert \\
\vm y^0-\xopt & \dots & \vm y^{p-1}-\xopt \\
\vert & & \vert 
\end{matrix}\right) \in \RR^{n\times p}.
\label{eq:BOBYQA_LGSmfn_Mlin}
\end{equation}
Further,
\begin{equation}
\vm v^{(k)} = 
\left(\begin{matrix}
v^{(k)}_{1}\\
\vdots\\
v^{(k)}_{p+n}
\end{matrix} \right) \in \RR^{p+n}, \ \ \label{eq:BOBYQA_LGSmfn_v}
\end{equation}
and
\begin{equation}
\bB = 
\left(\begin{matrix}
f(\vm y^0) - f(\xopt) \\
\vdots\\
f(\vm y^{p-1}) - f(\xopt)\\
\vm 0
\end{matrix} \right)
-\frac{1}{2}
\left(\begin{matrix}
\vm (\vm y^0-\xopt)^{\transpose}\vm H^{(k-1)} (\vm y^0-\xopt) \\
\vdots\\
\vm (\vm y^{p-1}-\xopt)^{\transpose}\vm H^{(k-1)} (\vm y^{p-1}-\xopt)\\
\vm 0
\end{matrix} \right)
\in \RR^{p+n}.
\label{eq:BOBYQA_LGSmfn_b}
\end{equation}
Then, after solving the interpolation problem~\eqref{eq:BOBYQA_LGSmfn}, the coefficients of~\eqref{eq:BOBYQA_QP} are computed by
\begin{align}
&c^{(k)} = f(\xopt), \ \ {\vm g^{(k)}}^{\transpose} = (v^{(k)}_{p+1}, \dots, v^{(k)}_{p+n})
\label{eq:BOBYQA_LGSmfn_c_g}\\
&\vm H^{(k)} = \vm H^{(k-1)} + \sum_{i=0}^{p-1} v^{(k)}_{i+1} \left((\vm y^i - \xopt)(\vm y^i - \xopt)^{\transpose}\right).
\label{eq:BOBYQA_LGSmfn_H}
\end{align}

\subsubsection{Convergence}\label{sec:BOBYQAConvergence}

The convergence theory for gradient based optimization algorithms like \gls{SQP} is typically based on error bounds of the Taylor expansion, in order to show the decreasing error between the model $m^{(k)}(\vm x)$ and the function $f(\vm x)$ and its derivatives. 
In \gls{DFO} the poisedness constant $\Lambda$ can be used to formulate a Taylor type error bound. In \cite{Conn_2008part1} an error bound is given by
\begin{eqnarray}
||\nabla m^{(k)}(\vm x) - \nabla f(\vm x)|| \leq \frac{1}{(d+1)!}G \Lambda \sum_{i=0}^{p}||\vm y^i - \vm x||^{d+1},
\label{eq:TaylorTypeErrorBound}
\end{eqnarray}
where $G$ is a constant depending only on the function $f$ and $d$ is the polynomial degree of the approximation model (i.e. here $d=2$). Thus, in order to apply the convergence theory of gradient based methods to \gls{DFO} methods, it is required to keep $\Lambda$ uniformly bounded for all training data sets used within the algorithm.

The algorithm in~\cite[Algo. 11.2]{Conn_2009book} is a modified version of the original BOBYQA such that global convergence can be proven. One key aspect is the fact, that the interpolation set is $\Lambda$-poised in each step -- or, within a final number of steps, it can be transformed into a $\Lambda$-poised set (using the so-called model improvement algorithm~\cite[Algo. 6.3]{Conn_2009book}). 


The original implementation by Powell~\cite{Powell_2009aa} and the modified Python implementation by Cartis et al.~\cite{Cartis_2019aa} cannot guarantee this. They apply strategies to update the training data set, which \textit{hopefully} reduce the poisedness constant $\Lambda$, but they do not provide bounds~\cite{Conn_2008part1}. 
They would require $\Lambda$-poisedness checks more often and re-evaluations of the whole training data set in situations of poorly balanced training data sets, i.e., an usage of the model improvement algorithm. But, for the benefit of less computational effort and thus, efficiency, they abdicate on provable convergence and rely on heuristics, when and how to check and improve $\Lambda$-poisedness.

\section{Hermite-type modifications for BOBYQA}
\label{sec:firststep}
As mentioned in sec.~\ref{sec:setting} we assume that we know some partial derivatives of the objective function $f:\RR^n\rightarrow\RR$, i.e., we can calculate them with negligible computational effort compared to the effort of evaluating the objective function itself. We focus on the information in~\eqref{eq:Setting_SetDerivs}, i.e., we neglect the second order derivatives.
Our \textit{Hermite} training data set is then given by
\begin{align}
\trainH = &\left\lbrace\left(\vm y^0,f(\vm y^0),\frac{\partial}{\partial y_1}f(\vm y^0), \dots, \frac{\partial}{\partial y_{\nd}}f(\vm y^0)\right),\dots \right. \label{eq:Mixed_TrainingDataSet} \\
&\left. \dots, \left(\vm y^p,f(\vm y^p),\frac{\partial}{\partial y_1}f(\vm y^p), \dots, \frac{\partial}{\partial y_{\nd}}f(\vm y^p)\right)\right\rbrace. \notag
\end{align}
We want to use this additional information in order to improve the quadratic model. 
In the following we will propose two different strategies:
\begin{enumerate}
	\item \textit{Hermite least squares} extends the simple interpolation~\eqref{eq:BOBYQA_LGSfullydet} with derivative information yielding least squares regression. First we introduce this approach starting with a uniquely solvable interpolation problem, i.e., the number of training data points $p_1$ coincides with the dimension of the basis $q_1$. Then, we allow to reduce the number of training data points such that the initial interpolation problem is underdetermined, i.e., $p_1<q_1$. While for $p_1=q_1$ global convergence can be ensured, the method yields superior performance for $p_1<q_1$.
	\item \textit{Hermite BOBYQA} extends BOBYQA's underdetermined interpolation~\eqref{eq:BOBYQA_LGSmfn} with derivative information yielding least squares regression.
\end{enumerate}

\subsection{Hermite least squares}
In Hermite interpolation a linear system is solved in order to find a polynomial approximation of a function, considering function values and partial derivative values in given training data points, cf.~\cite[Chap. 6.6]{Hermann_HermiteInterpol} or~\cite{Sauer_1995aa}. In the following we will build such a system, but with more information than required for a uniquely solvable Hermite interpolation and solve it with least squares regression. Thus, we call the optimization approach based on this technique Hermite least squares.

\subsubsection{Build upon interpolation ($p_1 = q_1$)}
First, we consider a training data set with $\abs{\trainI}=p_1 \equiv q_1$, i.e., we could solve an interpolation problem as in~\eqref{eq:BOBYQA_LGSfullydet} based only on function evaluations in order to obtain the quadratic model in~\eqref{eq:BOBYQA_QP}.
Instead, we provide additionally derivative information for the first $\kd$ partial derivatives of each training data point, i.e., we consider $\trainH$ with $\abs{\trainH} = p_1(1+\kd)$, where $p_1$ is the number of data points and $\abs{\trainH}$ denotes the \textit{number of information}. We extend the system with the gradient information available in form of additional lines for the system matrix and the right hand side and obtain
\begin{equation}
\Msr =
\left(\begin{matrix}
\phi_1(\vm y^0 - \vm x^{\text{opt}}) & \dots & \phi_q(\vm y^0- \vm x^{\text{opt}})  \\
\vdots & & \vdots\\
\phi_1(\vm y^{p-1}- \vm x^{\text{opt}}) & \dots & \phi_q(\vm y^{p-1}- \vm x^{\text{opt}}) \\
\frac{\partial}{\partial y_1}\phi_1(\vm y^0 - \vm x^{\text{opt}}) & \dots &
\frac{\partial}{\partial y_1}\phi_q(\vm y^0 - \vm x^{\text{opt}}) \\
\vdots & & \vdots\\
\frac{\partial}{\partial y_{\nd}}\phi_1(\vm y^{0} - \vm x^{\text{opt}}) & \dots &
\frac{\partial}{\partial y_{\nd}}\phi_q(\vm y^{0} - \vm x^{\text{opt}}) \\
\frac{\partial}{\partial y_1}\phi_1(\vm y^1 - \vm x^{\text{opt}}) & \dots &
\frac{\partial}{\partial y_1}\phi_q(\vm y^1 - \vm x^{\text{opt}}) \\
\vdots & & \vdots\\
\frac{\partial}{\partial y_{\nd}}\phi_1(\vm y^{p} - \vm x^{\text{opt}}) & \dots &
\frac{\partial}{\partial y_{\nd}}\phi_q(\vm y^{p} - \vm x^{\text{opt}}) 
\end{matrix}\right) \in \RR^{p_1(1+\nd)-1\times q}
\label{eq:Mixed1_Mregress}
\end{equation}
and
\begin{equation}
\bsr =
\left(\begin{matrix}
f(\vm y^0) - f(\vm x^{\text{opt}}) \\
\vdots \\ f(\vm y^{p-1}) - f(\vm x^{\text{opt}})
\\ \frac{\partial}{\partial y_1} f(\vm y^0) \\ \vdots
\\ \frac{\partial}{\partial y_{\nd}} f(\vm y^{0}) 
\\ \frac{\partial}{\partial y_1} f(\vm y^1) \\ \vdots
\\ \frac{\partial}{\partial y_{\nd}} f(\vm y^{p})
\end{matrix}\right) \in \RR^{p_1(1+\nd)-1}.
\label{eq:Mixed1_bregress}\\
\end{equation}
Solving the overdetermined linear system
\begin{equation}
\Msr \vsr \stackrel{\text{l.s.}}{=} \bsr
\label{eq:Mixed1_LGS}
\end{equation}
using least squares regression yields a quadratic model for the trust region subproblem ($\vsr$ defined as in~\eqref{eq:BOBYQA_LGSfullydet_v_b}). The formulation of the system matrix $\Msr$ and the right hand side $\bsr$ in case of second order derivatives is given in Appendix~\ref{sec:Hls_2ndDerivs}.
The question is: how good is this model. Therefore, we state the following theorem, which generalizes theorem 4.1 in~\cite{Conn_2008part2}.
\begin{theorem}
	\label{theo:LambdaPoisedInterpolRegress}
	Given a poised training data set $\train_{\text{\normalfont{I}}}$ and the monomial basis $\Phi$ with $\abs{\train_{\text{\normalfont{I}}}}=\abs{\Phi}$, and $\mathcal{B}\subset{\RR^n}$. Let $\vm M_{\text{\normalfont{I}}}$ be the corresponding system matrix of the interpolation problem and $\vm b_{\text{\normalfont{I}}}$ the right hand side, respectively. 
	Let $\train_{\text{\normalfont{R}}}\supset\train_{\text{\normalfont{I}}}$ be a training set containing further information, such that the corresponding system matrix $\vm M_{\text{\normalfont{R}}}$ has still full column rank.
	If $\train_{\text{\normalfont{I}}}$ is $\Lambda$-poised in $\mathcal{B}$ in the interpolation sense, then $\train_{\text{\normalfont{R}}}$ is at least $\Lambda$-poised in $\mathcal{B}$ in the regression sense.
\end{theorem}

\begin{proof}
	The additional information can be added in form of additional lines to the system matrix and the right hand side. Thus, we can set the system matrix and the right hand side of the regression problem corresponding to $\train_{\text{\normalfont{R}}}$ to
	\begin{equation}
	\vm M_{\text{R}} = 
	\left(\begin{matrix}
	\vm M_{\text{I}} \\
	\vm M_{\text{add}}
	\end{matrix}\right) \ \text{ and } \
	\vm b_{\text{R}} = 
	\left(\begin{matrix}
	\vm b_{\text{I}} \\
	\vm b_{\text{add}}
	\end{matrix}\right).
	\end{equation}
	Since $\train_{\text{\normalfont{I}}}$ is $\Lambda$-poised in the interpolation sense, by Def.~\ref{def:LambdaPoisedInterpol} holds
	\begin{equation}
	\forall \vm x \in \mathcal{B} \ \exists \vm{l}_{\text{I}}(\vm x) \in \RR^{\abs{\train_{\text{\normalfont{I}}}}} \text{ s.t. } \ \
	\sum_{\vm y^i \in \train_{\text{I}}} l_{\text{I}}^i(\vm x) 
	\vm m_{\text{I}}^i 
	= \vms \Phi(\vm x) \ \ \text{ with } \ \ \|\vm l_{\text{I}}(\vm x)\|_{\infty} \leq \Lambda, 
	\label{eq:Mixed1_LambdaPoisedInterpol}
	\end{equation}
	where $\vm m_{\text{I}}^i $ is the $i$-th column of $\vm M_{\text{I}}^{\transpose}$.
	We define $\vm l_{\text{R}}(\vm x) = (\vm l_{\text{I}}(\vm x), \vm 0)^{\transpose} \in \RR^{\abs{\train_{\text{\normalfont{R}}}}}$. Then 
	\begin{equation}
	\sum_{\vm y^i \in \train_{\text{R}}} l_{\text{R}}^i(\vm x) 
	\vm m_{\text{R}}^i 
	= \vms \Phi(\vm x) 
	\end{equation}
	holds and $\|\vm l_{\text{R}}(\vm x)\|_{\infty}$ is bounded by $\Lambda$, since
	\begin{equation}
	\|\vm \l_{\text{R}}(\vm x)\|_{\infty}
	= \max_{i=0,\dots,\abs{\train_{\text{R}}}} \abs{l_{\text{R}}^i(\vm x)}
	\stackrel{\text{Def. }\vm l_{\text{R}}}{=}
	\max_{i=0,\dots,\abs{\train_{\text{I}}}} \abs{l_{\text{I}}^i(\vm x)}
	= \|\vm l_{\text{I}}(\vm x)\|_{\infty}
	\stackrel{\eqref{eq:Mixed1_LambdaPoisedInterpol}}{\leq}\Lambda.
	\label{eq:Mixed_Proof_lambdaBounded}
	\end{equation}
\end{proof}

Theorem~\ref{theo:LambdaPoisedInterpolRegress} shows, that it is enough to ensure that a subset of the regression data set is $\Lambda$-poised in the interpolation sense, and then we can deduce that the regression data set is at least $\Lambda$-poised in the regression sense. Thus, 
although there is no analogue of the model improvement algorithm for the regression case~\cite[Chap. 6]{Conn_2009book},
we can apply the model improvement algorithm 6.3 for interpolation and the optimization algorithm 11.2 from~\cite{Conn_2009book}, ensure $\Lambda$-poisedness in the interpolation sense of a subset with $\abs{\Phi}=(n+1)(n+2)/2$ points, but build the quadratic model with least squares regression.
And since $l_{\text{R}}^i(\vm x)=0$ for $i>\abs{\train_{\text{I}}}$, the type of additional information in $\train_{\text{R}}$ has no impact on the proof -- as long as the matrix $\vm M_{\text{R}}$ has full column rank. This implies, that instead of additional data points and their function evaluations, we can also add derivative information according to~\eqref{eq:Mixed1_Mregress}, and our training data set remains at least $\Lambda$-poised.
The proof of convergence from~\cite{Conn_2009book} remains unaffected. In practice we expect faster convergence due to better quadratic models.

\subsubsection{Build upon underdetermined interpolation  ($p_1 < q_1$)}
For quadratic interpolation, a large set of training data points is required $p_1=\abs{\train_{\text{I}}}=(n+1)(n+2)/2$, such that the linear system is uniquely solvable. Since in Hermite least squares we have additional gradient information we can reduce the number of training data points and still have a determined or overdetermined regression system. The number of rows in the Hermite least squares system is given by $p_1(1+\kd)-1$ and has to be larger than the number of columns, i.e., $q=\abs{\Phi}-1=(n+1)(n+2)/2-1$, cf.~\eqref{eq:Mixed1_Mregress}. Thus, the required number of training data points in the Hermite least squares is only 
\begin{equation}
p_1 \geq \left\lceil \frac{(n+1)(n+2)}{2(1+\kd)} \right\rceil.
\label{eq:Mixed_Hermitian_NumberTrainPoints}
\end{equation}
The Hermite least squares system can be built as in~(\ref{eq:Mixed1_Mregress}--\ref{eq:Mixed1_LGS}), with the only difference that $p_1<q_1$. Since the regression data set does not contain a subset of $\Lambda$-poised interpolation points anymore, the model improvement algorithm cannot be applied to a subset and thus, convergence cannot be proven. In the next section we will discuss how to build and maintain the training data set, taking into account the derivative information.

\subsubsection{$\Lambda$-poisedness for Hermite least squares}
We aim to include the derivative information into the updating procedure of the training data set. 
We start with the Hermite interpolation setting, i.e., we set 
\begin{equation}
p_1 = \frac{(n+1)(n+2)}{2(1+\kd)}
\label{eq:Mixed_HermitianInterpol_NumberTrainPoints}
\end{equation}
s.t. $\abs{\trainH}=q_1$ and \eqref{eq:Mixed1_LGS} is a uniquely solvable Hermite interpolation problem.
We adapt the definition of $\Lambda$-poisedness to the Hermite interpolation case.
However, the following definition does not guarantee the required error bounds for provable convergence (cf.~\cite[Chap. 6.1]{Conn_2009book}).
This leads to an approach, without formal proof such as the common BOBYQA implementations.
\begin{definition}[$\Lambda$-poisedness in the Hermite interpolation sense]\hfill \\
	\label{def:LambdaPoisedHermiteInterpol}
	Given a poised Hermite interpolation problem as defined above with $p_1$ training data points, the training data set $\trainH$ and the monomial basis $\Phi$ with $\abs{\trainH}=p_1(1+\kd)=q_1=\abs{\Phi}$. Let $\mathcal{B}\subset \RR^n$ and $\Lambda>0$. Then the training data set $\trainH$ is $\Lambda$-poised in $\mathcal{B}$ (in the Hermite interpolation sense) if and only if
	\begin{equation}
	\forall \vm x \in \mathcal{B} \ \exists \vm{l}(\vm x) \in \RR^{q_1} \text{ s.t. } \ \
	\MH^{\top} \vm l(\vm x) = \vms \Phi(\vm x) \ \ \text{ with } \ \ ||\vm l(\vm x)||_{\infty} \leq \Lambda. 
	\label{eq:LambdaPoisedHermiteInterpol}
	\end{equation}
\end{definition}
We will define Lagrange-type polynomials for Hermite interpolation and show that they solve~\eqref{eq:LambdaPoisedHermiteInterpol}.
\begin{definition}[Lagrange-type polynomial for Hermite interpolation]\hfill \\
	\label{def:LagrPolyHermiteInterpol}
	Let $\MH\in \RR^{q_1 \times q_1}$ be a Hermite interpolation matrix with respect to the basis $\Phi$ as defined in~\eqref{eq:Mixed1_Mregress} and $\vm e^i \in \RR^{q_1}$ the $i$-th unit vector. Let $\vms \lambda^i$ solve
	\begin{equation}
	\MH \vms \lambda^i = \vm e^{i+1}.
	\label{eq:Mixed1_LagrangeTypePolyHermiteInterpol_LGS}
	\end{equation}
	Then, the polynomial built with the coefficients of $\vms \lambda^i$ and the basis $\Phi$
	\begin{equation}
	t^i(\vm x) = \lambda_{0}^i \phi_0(\vm x) + \dots + \lambda_{q}^i \phi_{q}(\vm x),
	\label{eq:Mixed1_LagrangeTypePolyHermiteInterpol}
	\end{equation}
	defines the $i$-th Lagrange-type polynomial for Hermite interpolation.
\end{definition}
\begin{lemma}
	Let $\vm t(\vm x) = \left( t^0(\vm x),\dots,t^q(\vm x) \right)^{\top}$ be defined as in~\eqref{eq:Mixed1_LagrangeTypePolyHermiteInterpol}, $\Phi(\vm x)$ defined as in sec.~\ref{sec:Notation} and $\MH$ from~\eqref{eq:Mixed1_Mregress}. Then, $\vm t(\vm x)$ solves $\MH^{\top} \vm l(\vm x) = \vms \Phi(\vm x)$, i.e., $\vm t(\vm x) \equiv \vm l(\vm x)$.
\end{lemma}
%
\begin{proof}
	We rewrite~\eqref{eq:Mixed1_LagrangeTypePolyHermiteInterpol_LGS} into
	\begin{equation}
	\MH \vm T = \vm I,
	\label{eq:Mixed1_LagrangeTypePolyHermiteInterpol_LGS_matrix}
	\end{equation}
	where $\vm I$ denotes the $q_1 \times q_1$ identity matrix and $\vm T$ is defined by the solution vectors of~\eqref{eq:Mixed1_LagrangeTypePolyHermiteInterpol_LGS}, i.e.,
	\begin{equation}
	\vm T = 
	\left(
	\begin{matrix}
	\vert & & \vert \\
	\vms \lambda^0 & \dots & \vms \lambda^q \\
	\vert & & \vert
	\end{matrix}
	\right).
	\label{eq:Mixed1_LagrangeTypePolyHermiteInterpol_T}
	\end{equation}
	Starting from~\eqref{eq:Mixed1_LagrangeTypePolyHermiteInterpol_LGS_matrix}, changing the order of the matrices, we derive
	\begin{align*}
	\vm T \MH &= \vm I. \\
	\intertext{Left multiplication by $\vms \Phi(\vm x)^{\transpose}$ yields}
	\vms \Phi(\vm x)^{\transpose} \vm T \MH &= \vms \Phi(\vm x)^{\transpose}.\\ 
	\intertext{We apply \eqref{eq:Mixed1_LagrangeTypePolyHermiteInterpol_T}}
	\left(\phi_0(\vm x), \dots, \phi_q(\vm x)\right) 
	\left( \begin{matrix}
	| & & | \\
	\vms \lambda^0 & \dots & \vms \lambda^q. \\
	| & & |
	\end{matrix} \right)
	\MH &= \left(\phi_0(\vm x), \dots, \phi_q(\vm x)\right) \\
	\intertext{and \eqref{eq:Mixed1_LagrangeTypePolyHermiteInterpol}}
	\left( t^0(\vm x),\dots,t^q(\vm x) \right) \MH &= \left(\phi_0(\vm x), \dots, \phi_q(\vm x)\right). \\
	\intertext{Transposing yields}
	\MH^{\transpose} 
	\left( \begin{matrix}
	t^0(\vm x) \\ \vdots \\ t^q(\vm x)
	\end{matrix} \right)
	&=
	\left( \begin{matrix}
	\phi_0(\vm x) \\ \vdots \\ \phi_q(\vm x)
	\end{matrix} \right), \\
	\intertext{which is per definition equivalent to}
	\MH^{\transpose} \vm t(\vm x) &= \vms \Phi(\vm x).
	\end{align*}
	Thus, $\vm t$ solves~\eqref{eq:LambdaPoisedHermiteInterpol}. 
\end{proof}
It holds that the (uniquely defined) polynomial solving the Hermite interpolation problem $\MH \vm v = \bH$ with $\MH\in \RR^{q_1 \times q_1}$ and $\bH \in \RR^{q_1}$ can be written as
\begin{equation}
m_{\text{H}}(\vm x) = 
\sum_{i=0}^{p} f(\vm y^i) t^i(\vm x) + \sum_{i=0}^{p} \sum_{j=1}^{\kd} 
\frac{\partial f}{\partial x_j}(\vm y^i) t^{jp_1-1+i}(\vm x),
\label{eq:LagrangeTypePolySum}
\end{equation}
where $p_1=p+1$ is the number of training data points and $q_1=(1+\kd)p_1$.

Let us investigate Hermite least squares, i.e., the case $\abs{\trainH}>q_1$. Extending the concept above, the Lagrange-type polynomials for Hermite least squares are obtained by solving
\begin{equation}
\MH \vms \lambda^i \stackrel{\text{l.s.}}{=} \vm e^{i+1}.
\label{eq:Mixed1_LagrangeTypePolyHermiteIRegress_LGS}
\end{equation}
instead of~\eqref{eq:Mixed1_LagrangeTypePolyHermiteInterpol_LGS}.
We maintain the training data set based on $\Lambda$-poisedness in the Hermite least squares sense. I.e., we maximize~\eqref{eq:BOBYQA_LeavingPoint} over the first $p_1$ Lagrange polynomials and replace the chosen data point with all corresponding information (i.e. function value and derivative information) by the new data point.

\subsection{Hermite BOBYQA}

The original BOBYQA method is characterized by the fact that it does not perform a fully determined interpolation, but uses $n+2\leq\abs{\train_{\text{B}}}<(n+1)(n+2)/2$ training data points. It guarantees existence and uniqueness by minimizing the difference between the current and the last Hessian matrix (in the Frobenius norm). Now we aim to improve this by using additional gradient information, i.e., we extend (\ref{eq:BOBYQA_LGSmfn}--\ref{eq:BOBYQA_LGSmfn_b}).

We take the coefficients from (\ref{eq:BOBYQA_LGSmfn_c_g}--\ref{eq:BOBYQA_LGSmfn_H}) and write the subscript B to emphasize the correspondence to the original BOBYQA method and obtain the original BOBYQA model by
\begin{equation}
m_{\text{B}}^{(k)}(\vm x) = c_{\text{B}}^{(k)} + {\vm g_{\text{B}}^{(k)}}^{\transpose}(\vm x - \vm x^{(k)}) 
+\frac{1}{2}(\vm x - \vm x^{(k)})^{\transpose}\vm H_{\text{B}}^{(k)} (\vm x - \vm x^{(k)}).
\label{eq:BOBYQA_QP_inModSec}
\end{equation}
We calculate the gradient of $m_{\text{B}}^{(k)}(\vm x)$ 
\begin{equation}
\nabla_{\vm x} m_{\text{B}}^{(k)}(\vm x) = \vm g_{\text{B}}^{(k)} + \vm H_{\text{B}}^{(k)} (\vm x - \vm x^{(k)})
\label{eq:BOBYQA_gradQP_inModSec}
\end{equation}
and replace $\vm H_{\text{B}}^{(k)}$ by~\eqref{eq:BOBYQA_LGSmfn_H}
\begin{equation}
\nabla_{\vm x} m_{\text{B}}^{(k)}(\vm x) 
= \vm g_{\text{B}}^{(k)} + \left(\vm H{^{(k-1)}} + \sum_{i=0}^{p-1} v_{\text{B},i+1}^{(k)} \left((\vm y^i - \xopt)(\vm y^i - \xopt)^{\transpose}\right)\right) (\vm x - \vm x^{(k)}).
\label{eq:BOBYQA_gradQP2_inModSec}
\end{equation}
Assume we know all partial derivatives, i.e., $\kd=n$. Then additionally to the interpolation conditions~\eqref{eq:BOBYQA_InterpolCond} we have
\begin{equation}
\nabla m^{(k)}(\vm y^j) = \nabla f(\vm y^j) \ \ \forall \vm y^j \in \trainH.
\label{eq:Mixed_InterpolCond}
\end{equation}
Inserting~\eqref{eq:BOBYQA_gradQP2_inModSec}, for each $\vm y^j \in \trainH$ we get
\begin{equation}
\sum_{i=0}^{p-1} v_{\text{B},i+1}^{(k)} 
\underbrace{\left((\vm y^i - \xopt)(\vm y^i - \xopt)^{\transpose}\right)}_{=:\vm C^i}
(\vm y^j - \vm x^{(k)}) + \vm g_{\text{B}}^{(k)}
= \nabla f(\vm y^j) - \vm H{^{(k-1)}}(\vm y^j - \vm x^{(k)}).
\label{eq:Mixed_InterpolCond2}
\end{equation}
Note that $\vm C^i \in \RR^{n\times n}$, so $\vm C^i(\vm y^j - \vm x^{(k)}) \in \RR^{n\times 1}$ and on both sides of~\eqref{eq:Mixed_InterpolCond2} we have vectors in $\RR^n$. 
For each $\vm y^j \in \trainH$ we formulate the linear system 
\begin{align}
&\underbrace{
	\left(\begin{matrix}
	| & & | & 1 & & 0\\
	\vm C^0(\vm y^j - \vm x^{(k)}) & \dots & \vm C^{p-1}(\vm y^j - \vm x^{(k)}) & & \ddots & \\
	| & & | & 0 & & 1
	\end{matrix}\right)
}_{=:\vm M_+^j}
\vm v^{(k)}
=\\
&\hspace{6cm} =\underbrace{
	\left(\begin{matrix}
	| \\ \nabla f(\vm y^j) - \vm H{^{(k-1)}}(\vm y^j - \vm x^{(k)}) \\ |
	\end{matrix}\right)
}_{=:\vm b_+^j}.\notag
\end{align}
%
and append it to the original BOBYQA system, s.t. we obtain
\begin{equation}
\MHB =
\left(\begin{matrix}
\MB \\ \vm M_+^0 \\ \vdots \\ \vm M_+^p
\end{matrix}\right)
\ \ \text{ and } \ \ 
\bHB =
\left(\begin{matrix}
\bB \\ \vm b_+^0 \\ \vdots \\ \vm b_+^p
\end{matrix}\right).
\end{equation}
In case only some partial derivatives are known, i.e., $\kd<n$, only the first $\nd$ rows of each $\vm M_+^j$ and $\vm b_+^j$ are included.
Since $\bB \in \RR^{p+n}$ and each $\vm b_+^j \in \RR^{\kd}$, we obtain $\bHB \in \RR^{p+n+p_1\kd}$, and $\MHB \in \RR^{p+n+p_1\kd\times p+n+p_1\kd}$ respectively. 
The Hermite BOBYQA system is then given by
\begin{equation}
\MHB \vm v^{(k)} \stackrel{\text{l.s.}}{=} \bHB
\label{eq:HermiteBOBYQA_LGS}
\end{equation}
and the coefficients for the Hermite BOBYQA model are computed analogously to the original BOBYQA method by (\ref{eq:BOBYQA_LGSmfn_c_g}--\ref{eq:BOBYQA_LGSmfn_H}).

In BOBYQA the update procedure for the training data set is based on the Lagrange polynomials in the minimum-norm sense and the corresponding Def.~\ref{def:LambdaPoisedMinNorm} for $\Lambda$-poisedness in the minimum-norm sense.
As in the Hermite least squares sense, we could use the same update procedure and argue that a subset of the regression set is $\Lambda$-poised (in the minimum-norm sense). However, the derivative information would not be considered when updating the training data set. Instead we solve
\begin{equation}
\MHB \vms \lambda^i \stackrel{\text{l.s.}}{=} \vm e^{i+1},
\label{eq:HermiteBOBYQA_Lagrange}
\end{equation}
build polynomials from the coefficients of $\vms \lambda^i$ according to~(\ref{eq:BOBYQA_LGSmfn_c_g}--\ref{eq:BOBYQA_LGSmfn_H}) and use them in the update decisions~(\ref{eq:BOBYQA_LeavingPoint}--\ref{eq:BOBYQA_NewPoint}).

\subsection{Scaling and weighting}

In this section we will discuss some preconditioning steps for solving the linear equations systems. In PyBOBYQA~\cite{Cartis_2019aa} the system is scaled in the following way: instead of
\begin{equation}
\vm M \vm v = \vm b
\label{eq:Scaling_LGSwithout}
\end{equation}
the system
\begin{equation}
\vm L \vm M \vm R \vm R^{-1} \vm v = \vm L \vm b
\label{eq:Scaling_LGSwith}
\end{equation}
is solved, where $\vm L$ and $\vm R$ are diagonal matrices of the same dimension as $\vm M$. Each training data point entry $\vm y^i$ is scaled by the factor $1/\Delta$, where $\Delta$ is the trust region radius of the current step (for simplicity of notation we omit the index $k$ for the current iteration for both the trust region radius and the system).
Thus, in the fully determined case the matrices are given by
\begin{equation}
\vm L = \vm I \ \text{ and } \ \vm R = 
\text{diag} \bigg(
\underbrace{\frac{1}{\Delta}  \ \ \dots \ \ \frac{1}{\Delta}}_{p} \ \
\underbrace{\frac{1}{\Delta^2} \ \ \dots \ \ \frac{1}{\Delta^2}}_{p-n}
\bigg),
\label{eq:Scaling_LR_fullydet}
\end{equation}
i.e., the columns for the linear part are scaled by $1/\Delta$, the columns for the quadratic part by $1/\Delta^2$.
In the underdetermined BOBYQA system the scaling matrices result in
\begin{equation}
\vm L = \vm R =
\text{diag} \bigg(
\underbrace{\frac{1}{\Delta^2}  \ \ \dots \ \ \frac{1}{\Delta^2}}_{p} \ \
\underbrace{\Delta \ \ \dots \ \ \Delta \vphantom{\frac{1}{\Delta^2}}}_{n}
\bigg).
\label{eq:Scaling_LR_BOBYQA}
\end{equation}
Preserving the same scaling scheme, the following left scaling matrices are obtained for the Hermite least squares approach
\begin{equation}
\vm L =
\text{diag} (
\underbrace{1  \ \ \dots \ \ 1}_{p} \ \
\underbrace{\Delta \ \ \dots \ \ \Delta}_{p_1\kd}
)
\label{eq:Scaling_LR_Hls}
\end{equation}
and 
\begin{equation}
\vm L =
\text{diag} \bigg(
\underbrace{\frac{1}{\Delta^2}  \ \ \dots \ \ \frac{1}{\Delta^2}}_{p} \ \
\underbrace{\Delta  \ \ \dots \ \ \Delta \vphantom{\frac{1}{\Delta^2}}}_{n} \ \
\underbrace{\frac{1}{\Delta} \ \ \dots \ \ \frac{1}{\Delta}}_{p_1\kd}
\bigg)
\label{eq:Scaling_LR_HBOBYQA}
\end{equation}
for the Hermite BOBYQA approach, respectively. The right scaling matrices remain unchanged, i.e., $\vm R$ as in~\eqref{eq:Scaling_LR_fullydet} for Hermite least squares and $\vm R$ as in~\eqref{eq:Scaling_LR_BOBYQA} for Hermite BOBYQA.


An optional step is the weighting of the least squares information, cf. weighted regression~\cite{Bjorck_1996}. Information belonging to a training data point close to the current solution can be given more weight, information belonging to a training data point far from the current solution can be given less weight. 
In the Hermite least squares approach this weighting can be applied to value and derivative information separately, because for each information there exists exactly one row. I.e., each row of the system matrix and the right hand side, which belongs to the data point $\vm y^i$, $i=0,\dots,p$ is multiplied with the corresponding weight. 
This is realized by the left multiplication of a diagonal weighting matrix $\vm D$ to both sides of~\eqref{eq:RegressProblemGeneral}, i.e., we obtain
\begin{equation}
\vm D \vm M \vm v \stackrel{\text{l.s}}{=} \vm D \vm b.
\label{eq:RegressProblemGeneral_weighted}
\end{equation}
In the Hermite BOBYQA approach, the first $p+n$ rows do not belong to one specific training data point. Thus, in this approach we only weight the derivative information, i.e., the first $p+n$ diagonal entries of $\vm D$ are set to $1$.

\section{Numerical tests}
\label{sec:NumericalTests}
A test set of almost $30$ nonlinear, bound constrained optimization problems with $2$, $3$, $4$, $5$ and $10$ dimensions has been evaluated and compared. The complete test set and the detailed results can be found at GitHub~\cite{Fuhrlander_Testset}. As reference solutions we consult the solution of PyBOBYQA using either fully determined interpolation, i.e., $p_1 = (n+1)(n+2)/2$, in the following referenced to as $\reffull$, or using the default setting in PyBOBYQA~\cite{Cartis_2019aa}, i.e., $p_1 = 2n+1$, in the following referenced to as $\refB$. For all remaining parameters we apply the default settings.
Another reference solution is \gls{SQP}, where the unknown derivatives are calculated with finite differences. Therefore, the SciPy implementation of the SLSQP algorithm from~\cite{Kraft_1988aa} has been used. Please note that this is a completely different implementation, thus a direct comparison should be handled with caution.
In the Hermite least squares approach, numerical tests showed that 
\begin{equation}
p_1 = \max\left(2n+1-\kd, \left\lceil \frac{(n+1)(n+2)}{2(1+\kd)} \right\rceil\right).
\label{eq:Num_Hermite_NumberTrainPoints}
\end{equation}
is a reasonable choice for the number of training data points. For the Hermite BOBYQA approach we use the same number of interpolation points as in $\refB$, 
i.e., $p_1 = 2n+1$.
Further, we vary the number of known derivative directions $\nd$ and always assume that these derivative directions are then available for all training data points.

In the tests we are interested in two aspects: 1) do we find an optimal solution and 2) how much computing time is needed. For 1) we check if we find the same solution as the reference methods. Please note that we considered only test functions for which the reference BOBYQA method was able to find the optimal solution. For 2) we compare the total number of objective function evaluations during the whole optimization process.

\subsection{Results}

Before we evaluate the complete test set, we analyze the accuracy of the quadratic model $m^{(k)}(\vm x)$ which is built in each iteration.
Let us consider the Rosenbrock function in $\RR^2$, given by
\begin{equation}
f(\vm x) = 100 (x_2 - x_1^2)^2 + (1 - x_1)^2.
\label{eq:Rosenbrock}
\end{equation}
We assume $\partial f / \partial x_1$ to be unknown and $\partial f / \partial x_2$ to be available.
The second order Taylor expansion $T_2f(\vm x;\vm x^{(k)})$ in $\vm x^{(k)}$ is considered as the reference model. We investigate the error between this Taylor reference and the quadratic model of $\refB$, Hermite least squares and Hermite BOBYQA evaluated by using the $L_2$-norm
\begin{equation}
\|m^{(k)}(\vm x) - T_2f(\vm x;\vm x^{(k)})\|_{L_2}^2
= \int_{\vm x^{(k)}-\delta}^{\vm x^{(k)}+\delta} \abs{m^{(k)}(\vm x) - T_2f(\vm x;\vm x^{(k)})}^2 \intd \vm x.
\label{eq:L2Norm}
\end{equation}
In Fig.~\ref{fig:Rosenbrock_ErrorQP}, the resulting error is plotted over the number of iterations. We observe that the error of the quadratic model decreases first for the Hermite least squares model. After $20$ iterations, the error remains below $2.5 \cdot 10^{-7}$. 
For $\refB$ and Hermite BOBYQA the error is also reduced to this magnitude, but it takes $65$ and $118$ iterations, respectively. However, the error for Hermite BOBYQA is already below  $10^{-6}$ after $50$ iterations, although at a slightly higher level than the other two methods. After an outlier in iteration $117$ it also reaches the error of $2.5 \cdot 10^{-7}$. The errors of the quadratic models reflect the performance of the three optimization methods. All of them find the same optimal solution. Hermite least squares is the most efficient method, it terminates after $43$ iterations. The reference method $\refB$ terminates after $106$ and Hermite BOBYQA after $153$ iterations. The number of objective function calls is proportional to the number of iterations.
\begin{figure}[t]
	\centering
	\includegraphics[width=0.6\textwidth]{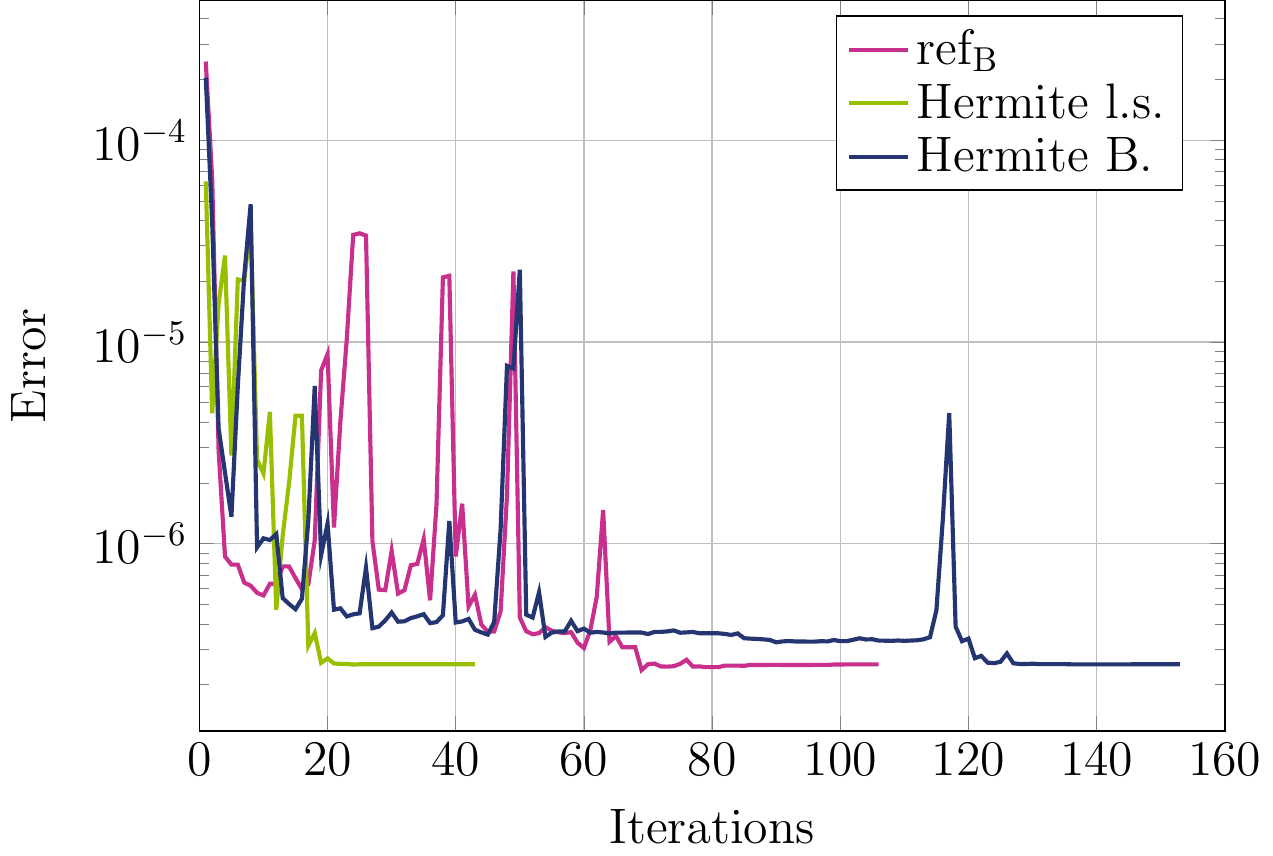}
	\caption{Error of the quadratic model $m^{(k)}$ in each iteration compared to second order Taylor expansion, cf.~\eqref{eq:L2Norm} with $\delta=0.01$. Three different methods are considered: classic BOBYQA ($\refB$), Hermite least squares and Hermite BOBYQA.}
	\label{fig:Rosenbrock_ErrorQP}
\end{figure}

\paragraph{Test set}
In the following, the test set from~\cite{Fuhrlander_Testset} is evaluated. The results are consistent with the observations regarding error and performance for the Rosenbrock function, which have been described above.
In Fig.~\ref{fig:results2dim}~--~\ref{fig:results10dim} the numerical results for the test set are visualized. We compare $\refB$ with Hermite least squares and Hermite BOBYQA and vary the number of known derivatives $\kd$, assuming that we know at least half of the derivative directions, i.e., $\kd\geq\frac{1}{2}n$. Then we take the average number of function evaluations. For example, in Fig.~\ref{fig:results3dim} for Hermite least squares with $\kd=2$ we average over all $3$-dimensional test problems, solved with Hermite least squares, with three cases each, i.e., 1) $\partial f/ \partial x_1$ and $\partial f/ \partial x_2$ are known, 2) $\partial f/ \partial x_1$ and $\partial f/ \partial x_3$ are known and 3) $\partial f/ \partial x_2$ and $\partial f/ \partial x_3$ are known. For the $10$-dimensional test problems we tested three random permutations of known derivatives per $\kd$. 
\begin{figure}[htb]
	\begin{minipage}[t]{.49\textwidth}
		\centering
		\includegraphics[width=\textwidth]{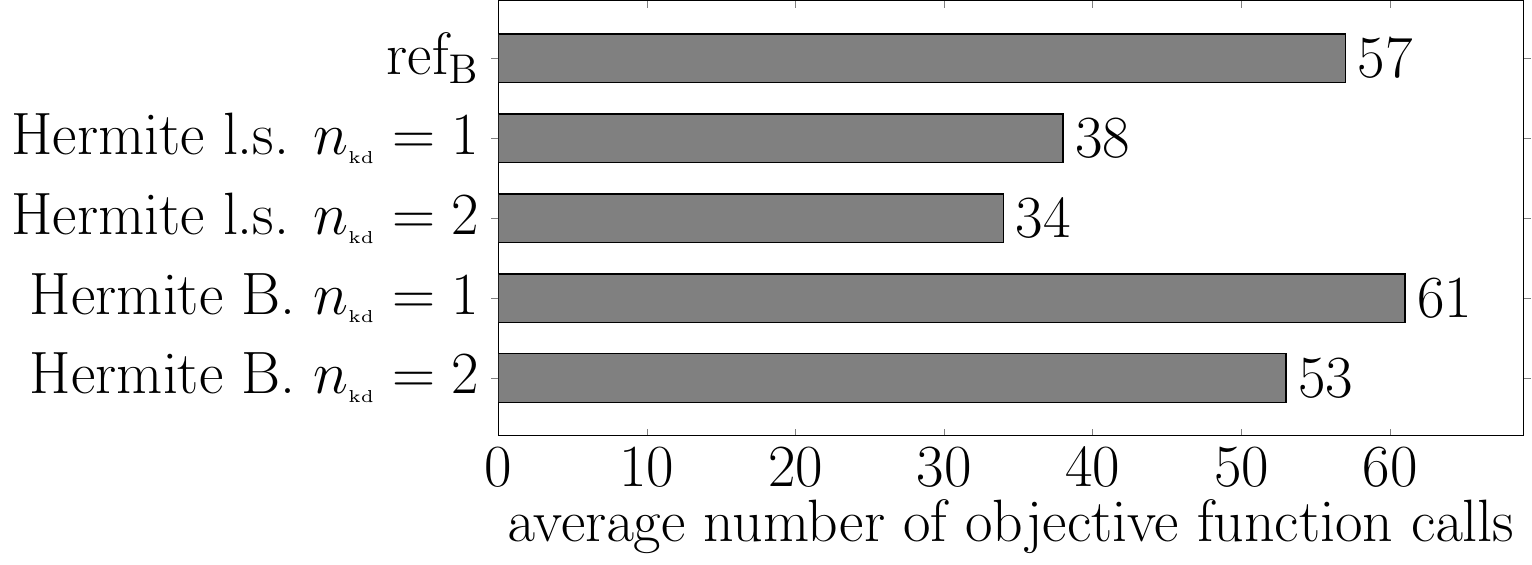}
		\subcaption{$2$-dimensional test problems.}\label{fig:results2dim}
	\end{minipage}
	\hfill
	\begin{minipage}[t]{.49\textwidth}
		\centering
		\includegraphics[width=\textwidth]{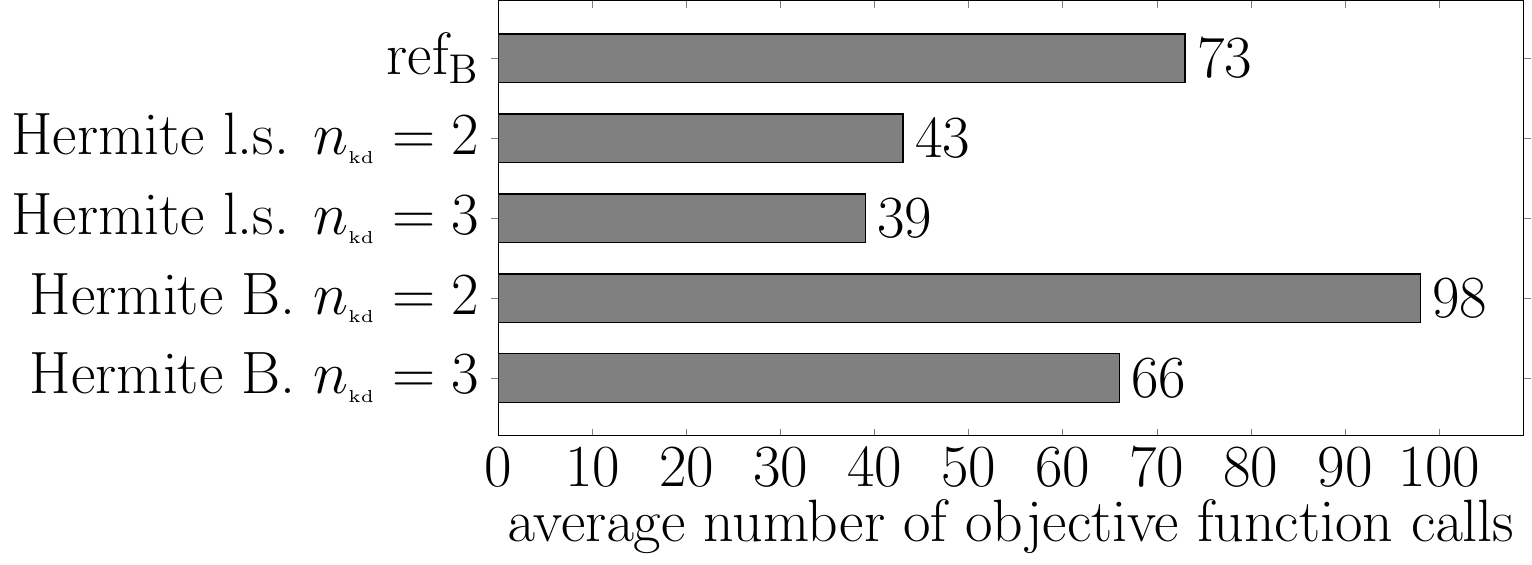}
		\subcaption{$3$-dimensional test problems.}\label{fig:results3dim}
	\end{minipage} \\ 
	\begin{minipage}[t]{.49\textwidth}
		\centering
		\includegraphics[width=\textwidth]{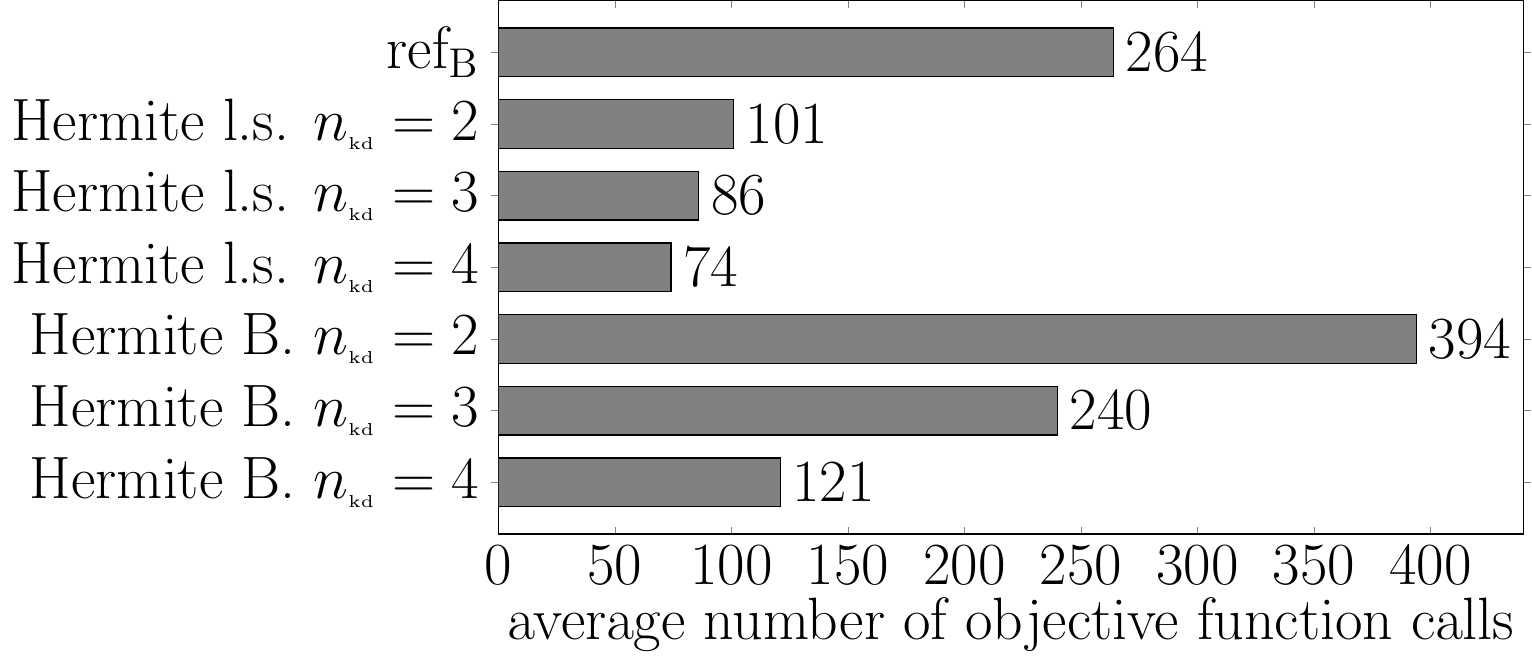}
		\subcaption{$4$-dimensional test problems.}\label{fig:results4dim}
	\end{minipage}
	\hfill
	\begin{minipage}[t]{.49\textwidth}
		\centering
		\includegraphics[width=\textwidth]{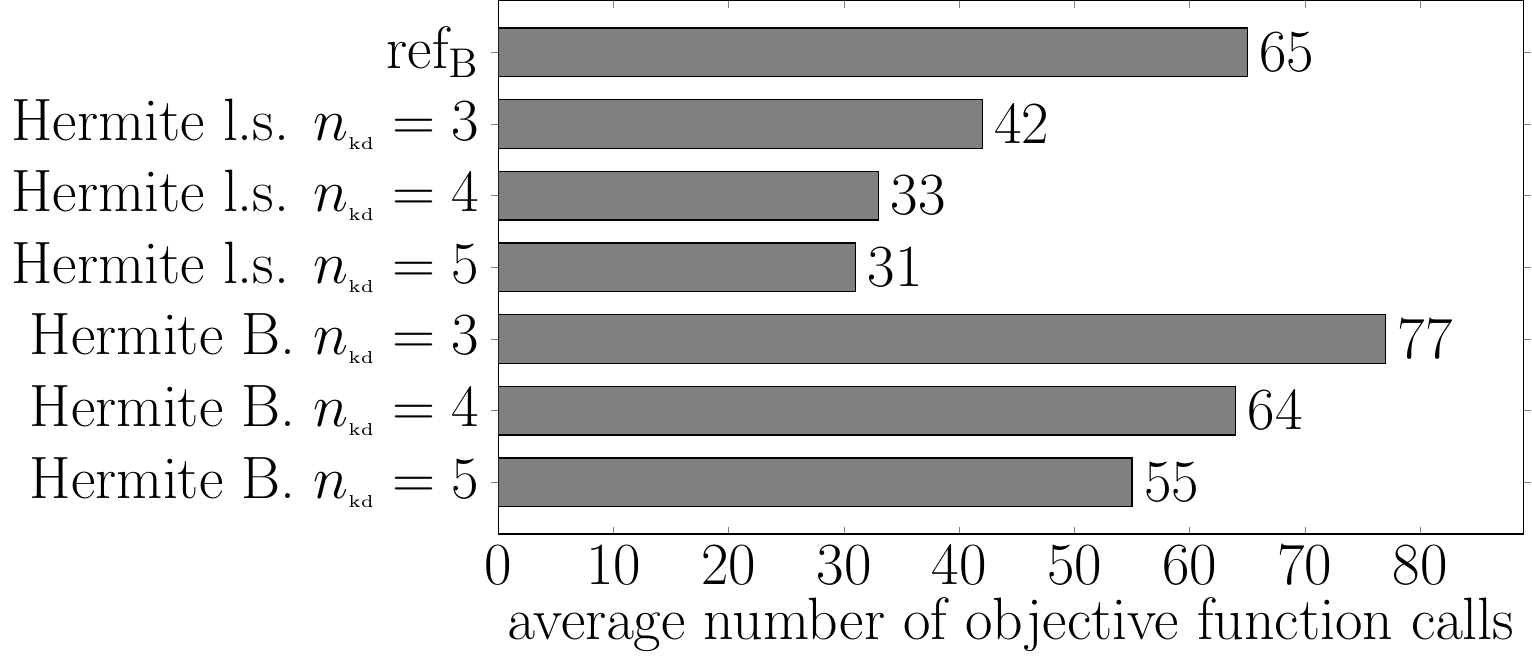}
		\subcaption{$5$-dimensional test problems.}\label{fig:results5dim}
	\end{minipage}\\ 
	\begin{minipage}[t]{.245\textwidth}
		\hspace{0.1cm}
	\end{minipage}
	\hfill
	\begin{minipage}[t]{.49\textwidth}
		\centering
		\includegraphics[width=\textwidth]{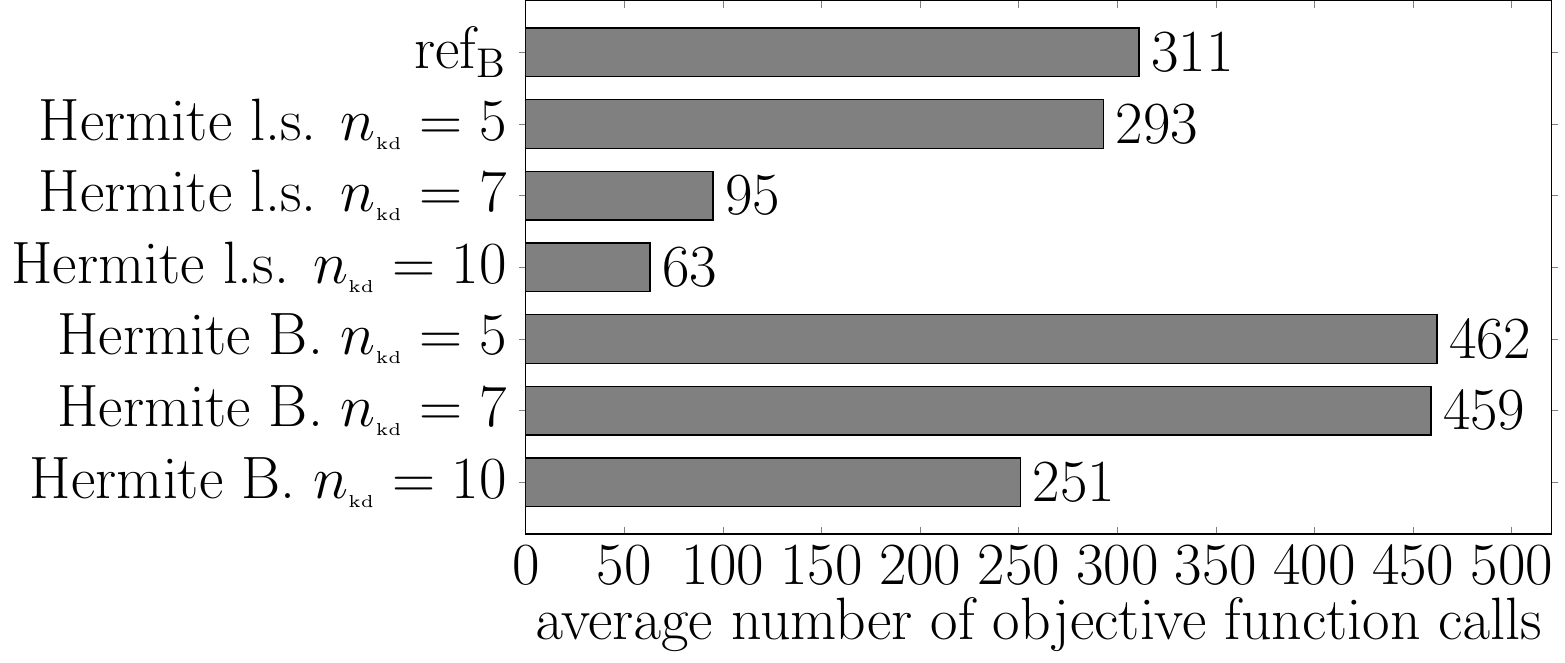}
		\subcaption{$10$-dimensional test problems.}\label{fig:results10dim}
	\end{minipage}
	\hfill
	\begin{minipage}[t]{.245\textwidth}
		\hspace{0.1cm}
	\end{minipage}
	\caption{Average number of function evaluations for all test problems solved with the reference BOBYQA method ($\refB$), Hermite least squares (Hermite l.s.) and Hermite BOBYQA (Hermite B.) with varying number of known derivatives $\kd$.}
	\label{fig:results}
\end{figure}

With Hermite least squares the number of objective function calls can be reduced by $34\,\% - 80\,\%$ compared to $\refB$, depending on dimension $n$ and number of known derivatives $\kd$. One exception is the case $n=10$ and $\kd=5$, where we can only see a reduction by $6\,\%$. Although, we observe that there are instances of some test problems for which the number of objective function calls increase (cf.~\cite{Fuhrlander_Testset}), in average the computational effort can be significantly reduced. The optimal solution has been found in all considered cases.
For the Hermite BOBYQA method we see that the computational effort increases for $\kd\approx \frac{1}{2}n$ by $8\,\% - 49\,\%$, while it can be decreased for $\kd\approx n$ by $7\,\% - 54\,\%$. In three test cases the optimal solution could not be found. The corresponding results have been excluded from the summary plots in Fig.~\ref{fig:results} but are listed in~\cite{Fuhrlander_Testset}.
For both variants, Hermite least squares and Hermite BOBYQA, we see that the more gradient information we have, the less objective function evaluations are needed within one optimization run. 

We can conclude that, assuming about the half or more of the partial derivatives are known, using the Hermite least squares approach instead of the classic BOBYQA method ($\refB$) reduces the computational effort significantly. On the other hand, the Hermite BOBYQA approach is only able to reduce the computational effort for a very high number of known derivatives. However, in case most of the derivatives are known, it is beneficial to use \gls{SQP} and calculate the few missing derivatives with finite differences.

In the numerical tests, we also compared the original BOBYQA and the proposed modifications with \gls{SQP}. While for Hermite least squares and Hermite BOBYQA we took PyBOBYQA from~\cite{Cartis_2019aa} as a basis and included the required modifications for the Hermite approaches, the \gls{SQP} method from SciPy based on~\cite{Kraft_1988aa} is a different implementation, using for example different ways to solve the quadratic subproblem. However, we could observe that in almost all cases, the \gls{SQP} method required a lower number of objective function calls in order to find the optimal solution than $\refB$, and in most cases also less than Hermite least squares. Please note that in the \gls{SQP} method we provided the known derivatives and only calculated the remaining ones with finite differences.


\paragraph{Weighted regression}
We apply an exponential weighting scheme to the Hermite least squares and the Hermite BOBYQA approach, where the weight for a training data point $\vm y^i$ is defined by
\begin{equation}
w(\vm y^i) = \frac{1}{\exp{(s)}} \exp{\left( s-s\frac{||\vm y^i-\vm x^{(k)}||_2}{\max_{i=0,\dots,p}||\vm y^i-\vm x^{(k)}||_2}\right)}, \ \ \forall i=0,\dots,p
\label{eq:Weighting_expo}
\end{equation}
with scaling factor $s>0$. In our tests we set $s=5$.
Please note that the number of weights is smaller than the dimension of $\vm D$ since several rows belong to the same training data point and thus, are weighted with the same $w(\vm y^i)$. 
For most of the test cases the weighting has no significant impact. 
We therefore do not apply the weighting per default.

\subsection{Noisy data}\label{sec:noisy}

Let us consider the Rosenbrock function 
from~\eqref{eq:Rosenbrock}
and investigate the performance of the different methods under noise.
In accordance to~\cite{Cartis_2019aa}, for that purpose we add random statistical noise to the objective function value and the derivative values by multiplying the results with the factor $1+\xi$, where $\xi$ is a uniformly distributed random variable, i.e., $\xi \sim \mathcal{U}(-10^{-2},10^{-2})$. Again, for \gls{SQP} and the Hermite approaches we assume $\partial f / \partial x_1$ to be unknown and $\partial f / \partial x_2$ to be available.

The optimal solution of~\eqref{eq:Rosenbrock} is
\begin{equation}
\vm x^{\text{opt}} = (1,1) \ \text{ with } \ f\left(\vm x^{\text{opt}}\right) = 0.
\label{eq:Rosenbrock_optsol}
\end{equation}
We start the optimization with
\begin{equation}
\vm x^{\text{start}} = (1.2,2) \ \text{ with } \ f\left(\vm x^{\text{start}}\right) = 31.4.
\label{eq:Rosenbrock_start}
\end{equation}
First we apply the Hermite least squares method. It terminates after only $37$ function calls with the optimal solution
\begin{equation}
\vm x^{\text{H.l.s.}} = (1,1) \ \text{ with } \ f\left(\vm x^{\text{H.l.s.}}\right) = 1.02 \cdot 10^{-23}.
\label{eq:Rosenbrock_optsol_Hls}
\end{equation}
Without noise, $43$ function calls were needed to find the optimum. The noise did not lead to an increase in computing effort.
Next, we apply Hermite BOBYQA, which terminates after $109$ function calls and also reaches the optimal solution
\begin{equation}
\vm x^{\text{H.B.}} = (1,1) \ \text{ with } \ f\left(\vm x^{\text{H.B.}}\right) = 2.97 \cdot 10^{-19}.
\label{eq:Rosenbrock_optsol_HB}
\end{equation}
We compare these results to the reference solution $\refB$. After $43$ objective function calls the algorithm terminates without reaching the optimum
\begin{equation}
\vm x^{\text{B}} = (1.41, 1.98) \ \text{ with } \ f\left(\vm x^{\text{B}}\right) = 0.16.
\label{eq:Rosenbrock_optsol_B}
\end{equation}
Additionally to $\refB$ we consider the PyBOBYQA version for noisy data from~\cite[Sec. 7]{Cartis_2019aa} as reference solution $\refN$. The main differences compared to $\refB$ are another choice of default parameters for adjusting the trust region radius (better suited for noisy data), sample averaging and multiple restarts.
Even with a high budget of $2000$ objective function calls the algorithm does not terminate. In order to compare the results with Hermite least squares, we set the budget to the number of required objective function calls to terminate the Hermite least squares method, i.e., to $37$, and evaluate $\refN$
\begin{equation}
\vm x^{\text{N},37} = (1.08, 1.17) \ \text{ with } \ f\left(\vm x^{\text{N},100}\right) = 0.01.
\label{eq:Rosenbrock_optsol_BN_100}
\end{equation}
This means, the optimum could not be sufficiently identified within this budget.
Finally, we apply the gradient based \gls{SQP}. It terminates \textit{successfully} after $41$ objective function calls, and even though the solution has improved, the optimum could not be reached
\begin{equation}
\vm x^{\text{SQP}} = (0.21, -0.00) \ \text{ with } \ f\left(\vm x^{\text{SQP}}\right) = 0.31.
\label{eq:Rosenbrock_optsol_SQP}
\end{equation}

The results are visualized in Fig.~\ref{fig:Rosenbrock_noisy}. We conclude that the gradient based solver \gls{SQP} fails as to be expected in optimizing the noisy Rosenbrock function. While the standard PyBOBYQA version $\refB$ also terminates without reaching the optimum, the noisy PyBOBYQA version $\refN$ reaches the optimum, but does not terminate.
The regression approach in the Hermite least squares and the Hermite BOBYQA method robustify the optimization under noisy data. Both achieve the optimal solution at low computational costs, especially the Hermite least squares requires only $37$ objective function calls.
\begin{figure}[t]
	\centering
	\includegraphics[width=0.7\textwidth]{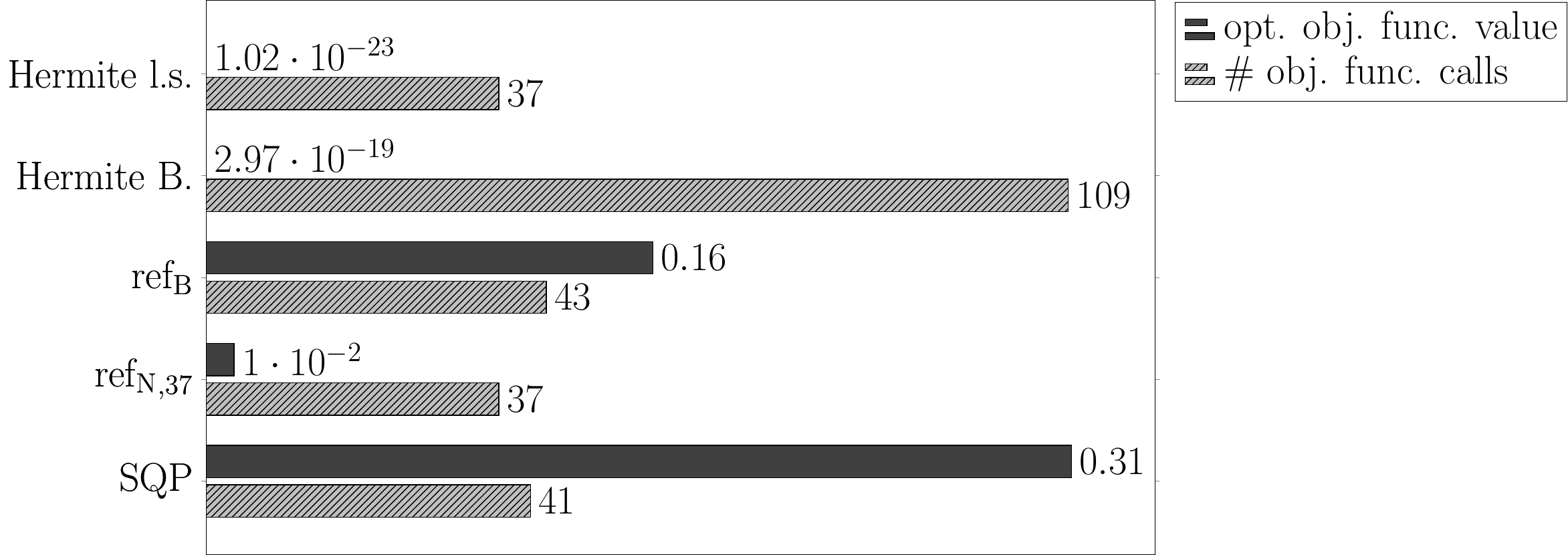}
	\caption{Results for the optimization of the noisy Rosenbrock function~\eqref{eq:Rosenbrock}, solved with Hermite least squares (Hermite l.s.), Hermite BOBYQA (Hermite B.), the reference BOBYQA method ($\refB$), the reference BOBYQA method for noisy data with maximum budget $37$ (${\refN}_{,37}$) and the reference SQP method (SQP).}
	\label{fig:Rosenbrock_noisy}
\end{figure}

\subsection{Hermite least squares with second order derivatives}\label{sec:Num2ndderivs}

Again we consider the Rosenbrock function~\eqref{eq:Rosenbrock} as test function and investigate the usage of second order derivatives, according to the formulation in Appendix~\ref{sec:Hls_2ndDerivs}. In this example we observe that the usage of second derivatives additionally to first derivatives further reduces the computing effort. The results are given in Table~\ref{tab:2ndDerivsResults}. Since second order derivatives are rarely available in practice, so we do not further extend the numerical tests.
\begin{table}[h]
	\centering
	\begin{tabular}{c|ccc}
		& \multicolumn{3}{|c}{$\mathcal{I}_{\text{d}}$ resp. $\mathcal{I}_{\text{2d}}$} \\ 
		&  $\left[0\right]$     &  $\left[1\right]$  &  $\left[0,1\right]$   \\
		\hline 
		1st order         &  $67$        &  $43$     &  $40$   \\
		2nd order         &  $62$        &  $40$     &  $38$   
	\end{tabular}
	\caption{Number of objective function calls for Rosenbrock function~\eqref{eq:Rosenbrock} depending on the set of known derivatives and the usage of first and second order derivatives or first order derivatives only.}
	\label{tab:2ndDerivsResults}
\end{table}

\section{A practical example: yield optimization}
\label{sec:PracticalExample}

In this section we discuss a practical example where the case of known and unknown gradients occur. In the field of uncertainty quantification, yield optimization is a common task~\cite{Graeb_2007aa}. In the design process of a device, e.g. antennas, electrical machines or filters, geometry and material parameters are chosen such that predefined requirements are fulfilled. However, in the manufacturing process, there are often uncertainties which lead to a deviation in the optimized design parameters and this may cause a violation of the requirements. The aim of yield estimation is the quantification of the impact of this uncertainty. The yield defines the probability, that the device still fulfills the performance requirements, under consideration of the manufacturing uncertainties. Thus, the natural goal is to maximize the yield. Please note, the task of yield maximization is equivalent to the task of failure probability minimization. We will formally introduce the yield and discuss the task of yield optimization with an example from the field of electrical engineering: a simple dielectrical waveguide as depicted in Fig.~\ref{fig:waveguide}. The model of the waveguide originates from~\cite{Loukrezis_WG}, and was used for yield optimization previously, e.g. in~\cite{Fuhrlander_2020aa}.
\begin{figure}
	\centering
	\includegraphics[width=0.5\textwidth]{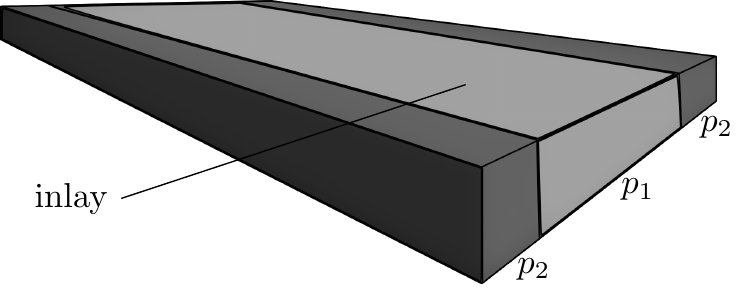}
	\caption{Model of a simple waveguide with dielectrical inlay and two geometry parameters $p_1$ and $p_2$.}
	\label{fig:waveguide}
\end{figure}

The waveguide has four design parameters, which shall be modified. Two uncertain geometry parameters: the length of the inlay $p_1$ and the length of the offset $p_2$. And two deterministic material parameters: $d_1$ with impact on the relative permittivity and $d_2$ with impact on the relative permeability. The uncertain parameters are modeled as Gaussian distributed random variables. Let the mean value (for the starting point, $k=0$) and the standard deviation been given by
\begin{equation}
\overline{p}_{1}^{(0)} = 9\,\text{mm}, \ \overline{p}_{2}^{(0)} = 5\,\text{mm}, \ \sigma_1 = \sigma_2 = 0.7.
\label{eq:yield_uncertpara}
\end{equation}
The starting points for the deterministic variables are
\begin{equation}
d_{1}^{(0)} = d_{2}^{(0)} = 1.
\label{eq:yield_determpara}
\end{equation}
The optimization variable is defined by
\begin{equation}
\vm x = (\overline{p}_{1}, \overline{p}_{2}, d_1, d_2)^{\transpose}.
\label{eq:yield_optvariable}
\end{equation}
As quantity of interest we consider the scattering parameter $S_r$ (S-parameter), which gives us information about the reflection behavior of the electromagnetic wave passing the waveguide. In order to calculate the value of the S-parameter for a specific setting, the electric field formulation of Maxwell has to be solved numerically, e.g. with the finite element method (FEM). Then we formulate the performance requirement by
\begin{equation}
S_r(\vm x) \leq -24 \, \text{dB} \ \forall r \in T_r=[2\pi 6.5,2\pi 7.5] \text{ in GHz},
\label{eq:yield_pfs}
\end{equation}
where the so-called range parameter $r$ is the angular frequency.
The range parameter interval $T_r$ is discretized in eleven equidistant points and~\eqref{eq:yield_pfs} has to be fulfilled for each of these points. The safe domain is the set of combinations of the uncertain parameters fulfilling the requirements, and depends on the current deterministic variable, i.e.,
\begin{equation}
\SD \equiv \SD_{d_1,d_2}(p_1,p_2) := \left\lbrace (p_1,p_2): S_r(\vm x) \leq -24 \, \text{dB} \ \forall r \in T_r \right\rbrace.
\label{eq:yield_SD}
\end{equation}

We follow the definitions from~\cite{Graeb_2007aa}. The yield, i.e., the probability of fulfilling all requirements~\eqref{eq:yield_pfs} under consideration of the uncertainties~\eqref{eq:yield_uncertpara}, is defined by
\begin{equation}
Y(\vm x) := \EE\left[\One_{\SD}(p_1,p_2)\right] = \int_{\RR}  \int_{\RR} \One_{\SD}(p_1,p_2) \pdf_{\overline{p}_1,\overline{p}_2,\sigma_1,\sigma_2} (p_1,p_2) \intd p_1 \intd p_2,
\label{eq:yield_yield}
\end{equation}
where $\One_{\SD}(p_1,p_2)$ defines the indicator function with value $1$ if $(p_1,p_2)$ lies inside the safe domain and $0$ elsewise, and $\pdf$ defines the probability density function of the two dimensional Gaussian distribution.
This can be numerically estimated by a Monte Carlo analysis, i.e.,
\begin{equation}
Y_{\text{MC}}(\vm x) = \frac{1}{\Nmc} \sum_{i=1}^{\Nmc} \One_{\SD}\left(p_1^{(i)},p_2^{(i)}\right),
\label{eq:yield_MC}
\end{equation}
where $(p_1^{(i)},p_2^{(i)})_{i=1,\dots,\Nmc}$ are sample points according to the distribution of the uncertain design parameters.
Since for each sample point the S-parameter has to be calculated, the yield estimator is a computationally expensive function. In the next step, this function shall be optimized, i.e.,
\begin{equation}
\max_{\vm x} Y(\vm x).
\label{eq:yield_opt}
\end{equation}
Since $\overline{p}_1$ and $\overline{p}_2$ in~\eqref{eq:yield_yield} are only contained in the probability density function the derivative with respect to the mean values of the uncertain parameters is given by
\begin{equation}
\frac{\partial}{\partial \overline{p}_j} Y(\vm x) = \int_{\RR}  \int_{\RR} \One_{\SD}(p_1,p_2) \frac{\partial}{\partial \overline{p}_j} \pdf_{\overline{p}_1,\overline{p}_2,\sigma_1,\sigma_2} (p_1,p_2) \intd p_1 \intd p_2, \ j=1,2.
\end{equation}
And since the probability density function of the Gaussian distribution is an exponential function, it is continuously differentiable, thus the derivatives with respect to $\overline{p}_1$ and $\overline{p}_2$ can be calculated easily. Further, according to~\cite{Graeb_2007aa}, the derivative of the MC yield estimator with respect to the uncertain parameters is given by
\begin{equation}
\frac{\partial}{\partial \overline{p}_j} Y_{\text{MC}}(\vm x) 
= Y_{\text{MC}}(\vm x) \frac{1}{\sigma_j^2} (\overline{p}_{j,\SD}-\overline{p}_j),  \ j=1,2.
\end{equation}
where $\overline{p}_{j,\SD}$ is the mean value of all sample points of $p_j$ lying inside the safe domain.
This implies that there are not only closed-form expressions of these derivatives, but also numerical expressions which require only the evaluation of the objective function (which is anyway necessary), but no further computational effort. 
On the other hand, the deterministic variables are contained in the indicator function in~\eqref{eq:yield_yield}. Thus, the corresponding partial derivatives are not considered as available. 
This leads to the situation that two partial derivatives are available, and two are unknown. The Hermite approaches described above can be applied and compared with standard BOBYQA and \gls{SQP}.

There are two possibilities for the generation of the Monte Carlos sample set: a) the same sample set is used in each iteration and just shifted according to the current mean value (no noise) and b) the sample set is generated newly each time the mean value is changed (noise). Depending on the size of the sample set, the accuracy can be controlled. In the following we investigate three different settings:
\begin{enumerate}
	\item no noise: same sample set (a) and $\Nmc=2500$
	\item low noise: new sample sets (b) and $\Nmc = 2500$
	\item high noise: new sample sets (b) and $\Nmc = 100$
\end{enumerate}
We start with the no noise setting and compare the different optimization methods with respect to the optimal value reached and the number of objective function calls needed. The initial yield value is $Y_{\text{MC}}^{(0)}=42.8\,\%$. The results are visualized in Fig.~\ref{fig:yield_nn}.
\begin{figure}
	\centering
	\includegraphics[width=0.7\textwidth]{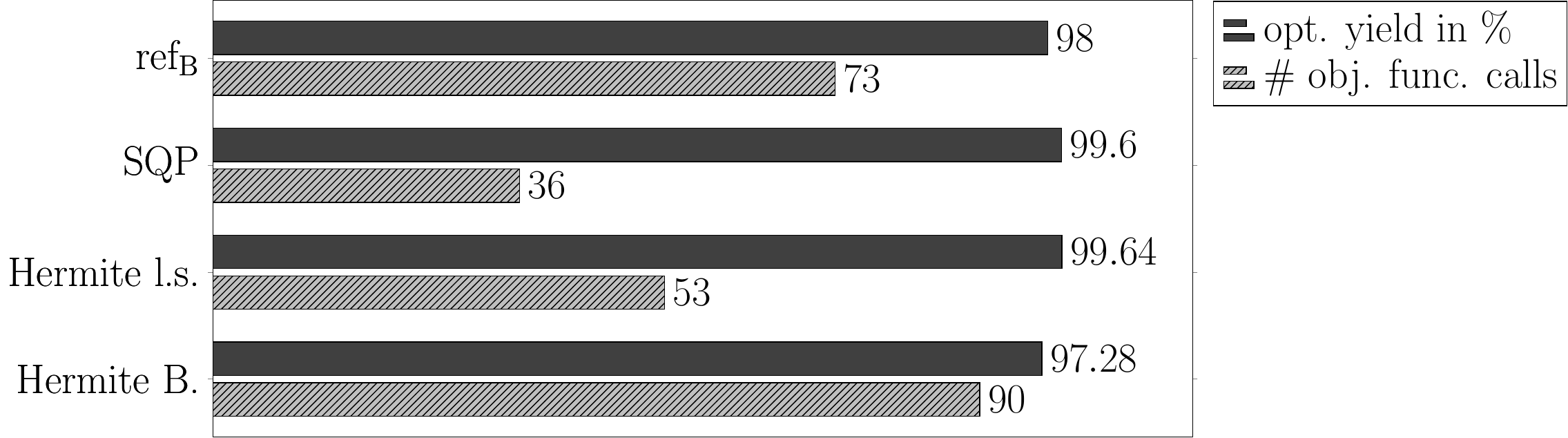}
	\caption{Results for yield optimization in the no noise setting, solved with the reference BOBYQA method ($\refB$), the reference SQP method (SQP), Hermite least squares (Hermite l.s.) and Hermite BOBYQA (Hermite B.).}
	\label{fig:yield_nn}
\end{figure}
While the optimal yield values are similar (best for Hermite l.s. and \gls{SQP}), the computational effort varies significantly. \gls{SQP} performs best with only $36$ objective function calls, Hermite l.s. is at the second position with $~50\,\%$ more, then $\refB$ with $~100\,\%$ more. Hermite BOBYQA performs worst with $90$ function calls, which corresponds to additional $150\,\%$. This coincides with our findings in sec.~\ref{sec:NumericalTests}. There we could also observe that Hermite least squares performes best (excluding \gls{SQP}) and Hermite BOBYQA requires more than $n/2$ known derivatives to be competitive with the other approaches.

In the next step we evaluate the noisy settings. 
The results for the low noise setting are visualized in Fig.~\ref{fig:yield_ln}, and for the high noise setting in Fig.~\ref{fig:yield_hn}, respectively.
\begin{figure}
	\begin{minipage}[t]{\textwidth}
		\centering
		\includegraphics[width=0.7\textwidth]{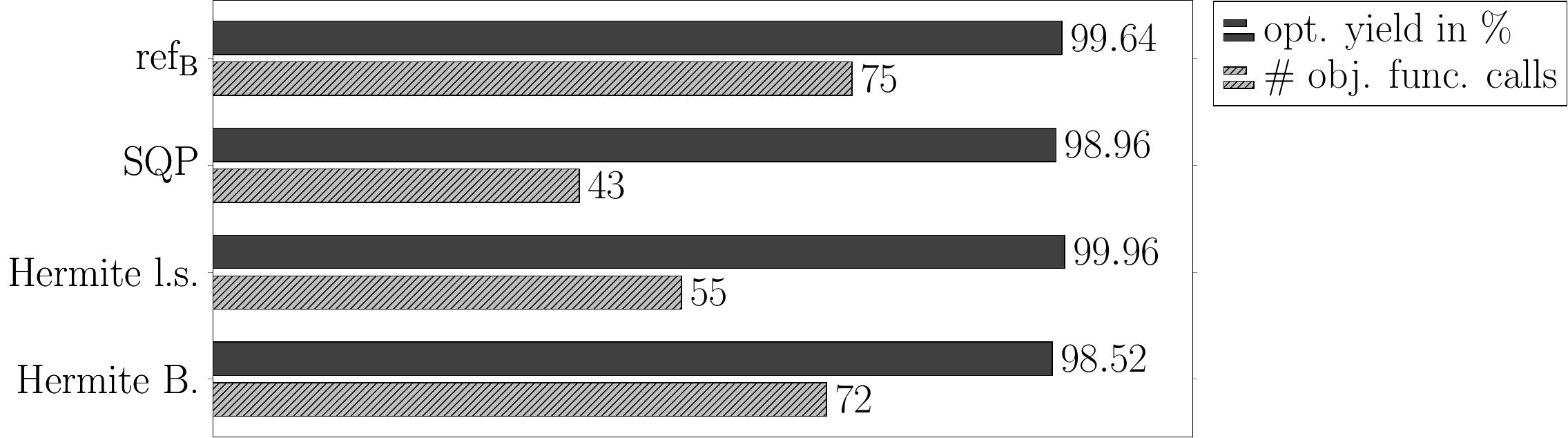}
		\subcaption{Low noise setting.}\label{fig:yield_ln}
	\end{minipage} \\ 
	\begin{minipage}[t]{\textwidth}
		\centering
		\includegraphics[width=0.7\textwidth]{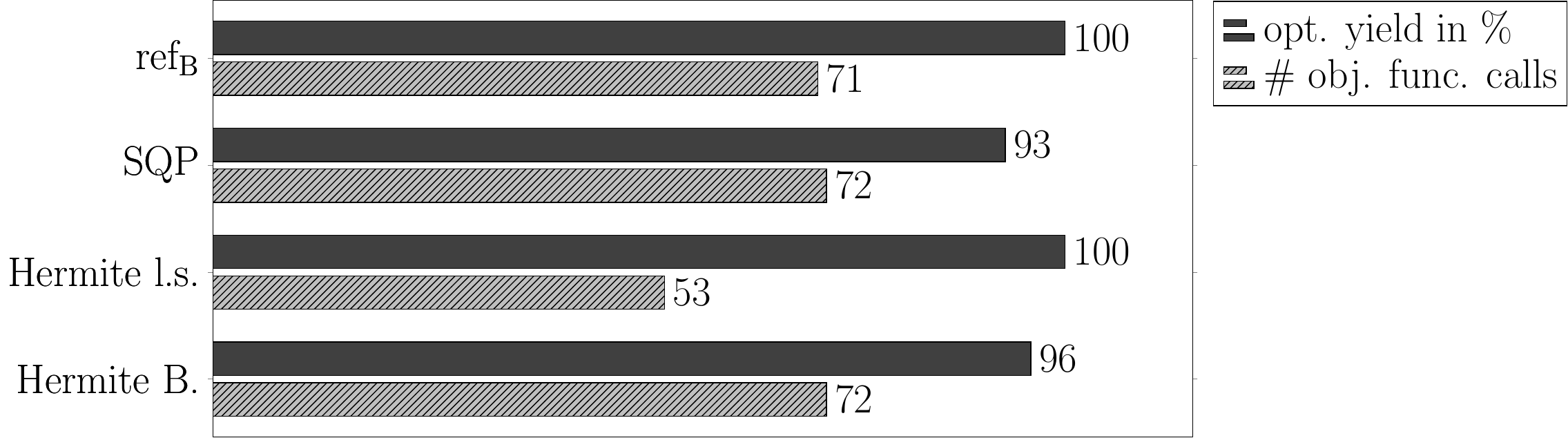}
		\subcaption{High noise setting.}\label{fig:yield_hn}
	\end{minipage}
	\caption{Results for yield optimization for noisy settings, solved with the reference BOBYQA method ($\refB$), the reference SQP method (SQP), Hermite least squares (Hermite l.s.) and Hermite BOBYQA (Hermite B.).}
	\label{fig:yield_noisy}
\end{figure}
%
As in the no noise setting, $\refB$ and the Hermite least squares find the optimal solution and the number of objective function calls does not change significantly when noise is added. With Hermite BOBYQA, the computational effort is reduced from $90$ function evaluations to $72$ in the noisy settings. Compared to the previous methods, the algorithm terminates with only slightly worse optimal solutions. 
While in the low noise setting \gls{SQP} finds a good optimal solution with the lowest number of objective function calls, in the high noise setting \gls{SQP} loses its advantage in terms of computational effort, and at the same time breaks down in finding an optimum.

%
In summary, for this practical example, the Hermite least squares method performs best in terms of solution quality and computational effort. Further we observe that the interpolation and regression based methods can handle the noise, while \gls{SQP} may not find the optimum anymore.

\section{Conclusion}
\label{sec:conclusion}

In this paper, we address the issue that in an optimization problem, some partial derivatives of the objective function are available and others are not. Based on Powell's derivative-free solver BOBYQA, we have developed two Hermite-type modifications: Hermite least squares and Hermite BOBYQA. In the Hermite least squares approach, besides function evaluations the available first and second order derivatives of a training data set is used to build a quadratic approximation of the original objective function. In each iteration, this quadratic subproblem is solved in a trust region and the training data set is updated. In Hermite BOBYQA, the regression problem builds upon BOBYQA's underdetermined interpolation system.

Under some assumptions, global convergence of the optimization algorithms can be proven. However, some small adjustments lead to higher performance regarding the computational effort and thus to higher practical applicability. Numerical tests on $30$ test problems including a practical example in the field of yield optimization have been performed. If half or more partial derivatives are available, the Hermite least squares approach outperforms classic BOBYQA in terms of computational effort by maintaining the ability of finding the optimal solution. Depending on the dimension and the amount of known derivative directions the number of objective function calls can be reduced by a factor up to five. If most of the partial derivatives are available, Hermite BOBYQA requires slightly less function calls than classic BOBYQA, however, its performance can be improved by applying exponential weighting. Both proposed methods are particularly stable with respect to noisy objective functions. In case of noisy data Hermite least squares and Hermite BOBYQA find the optimal solution more reliably and quickly than BOBYQA or gradient based solvers as sequential quadratic programming (SQP) using finite differences for missing derivatives.


\begin{appendices}

\section{Hermite least squares with second order derivatives} \label{sec:Hls_2ndDerivs}

The usage of second order derivatives is straightforward.
For the sake of completeness, we formulate the Hermite least squares system analog to~(\ref{eq:Mixed1_Mregress}--\ref{eq:Mixed1_LGS}). 
We denote the matrix with the second order derivative information corresponding to the $i$-th training data point by $\vm M_{\text{2d},i}$, the corresponding right hand side $\vm b_{\text{2d},i}$, respectively. They are given by
\begin{equation}
\vm M_{\text{2d},i} =
\left(\begin{matrix}
\frac{\partial^2}{\partial x_1 \partial x_1}\phi_1(\vm y^i - \vm x^{\text{opt}}) & \dots & \frac{\partial^2}{\partial x_1 \partial x_1}\phi_q(\vm y^i - \vm x^{\text{opt}}) \\
\frac{\partial^2}{\partial x_1 \partial x_2}\phi_1(\vm y^i - \vm x^{\text{opt}}) & \dots & \frac{\partial^2}{\partial x_1 \partial x_2}\phi_q(\vm y^i - \vm x^{\text{opt}}) \\
\vdots & & \vdots \\
\frac{\partial^2}{\partial x_1 \partial x_{\kds}}\phi_1(\vm y^i - \vm x^{\text{opt}}) & \dots & \frac{\partial^2}{\partial x_1 \partial x_{\kds}}\phi_q(\vm y^i - \vm x^{\text{opt}}) \\
\frac{\partial^2}{\partial x_2 \partial x_2}\phi_1(\vm y^i - \vm x^{\text{opt}}) & \dots & \frac{\partial^2}{\partial x_2 \partial x_2}\phi_q(\vm y^i - \vm x^{\text{opt}}) \\
\vdots & & \vdots \\
\frac{\partial^2}{\partial x_{\kds} \partial x_{\kds}}\phi_1(\vm y^i - \vm x^{\text{opt}}) & \dots & \frac{\partial^2}{\partial x_{\kds} \partial x_{\kds}}\phi_q(\vm y^i - \vm x^{\text{opt}}) 
\end{matrix}\right)
\label{2nd_M}
\end{equation}
with $\vm M_{\text{2d},i}\in \RR^{(\kds^2+\kds)/2 \times q}$ and
\begin{equation}
\vm b_{\text{2d},i} =
\left(\begin{matrix}
\frac{\partial^2}{\partial x_1 \partial x_1}f(\vm y^i)  \\
\frac{\partial^2}{\partial x_1 \partial x_2}f(\vm y^i)  \\
\vdots  \\
\frac{\partial^2}{\partial x_1 \partial x_{\kds}}f(\vm y^i)  \\
\frac{\partial^2}{\partial x_2 \partial x_2}f(\vm y^i)  \\
\vdots  \\
\frac{\partial^2}{\partial x_{\kds} \partial x_{\kds}}f(\vm y^i) 
\end{matrix}\right)
\in \RR^{(\kds^2+\kds)/2}.
\label{2nd_b}
\end{equation}

Utilizing that the basis $\Phi$ is defined by~\eqref{eq:MonomialBasis2} and assuming the second order derivatives are available for all directions, i.e., $\kds=n$, $\vm M_{\text{2d},i}$ can be simplified to
\begin{equation}
\vm M_{\text{2d},i}^{\text{simpl.}} =
\left(\begin{matrix}
0 & \dots & 0 & 1 & & \\
\vdots & & \vdots & & \ddots & \\
0 & \dots & 0 &  & & 1
\end{matrix}\right),
\label{2nd_M_simple}
\end{equation}
i.e., the linear part vanishes and the quadratic part is given by the identity matrix.

\section*{Acknowledgments}
The work of Mona Fuhrländer is supported by the Graduate School CE within the Centre for Computational Engineering at Technische Universität Darmstadt.




\end{appendices}



\begin{thebibliography}{10}
	
	\bibitem{Audet_2017aa}
	{\sc C.~Audet and W.~Hare}, {\em Derivative-Free and Blackbox Optimization}, 01
	2017, \url{https://doi.org/10.1007/978-3-319-68913-5}.
	
	\bibitem{Billups_2013}
	{\sc S.~C. Billups, J.~Larson, and P.~Graf}, {\em Derivative-free optimization
		of expensive functions with computational error using weighted regression},
	SIAM Journal on Optimization, 23 (2013), pp.~27--53.
	
	\bibitem{Bjorck_1996}
	{\sc {\AA}.~Bj{\"o}rck}, {\em Numerical methods for least squares problems},
	SIAM, 1996.
	
	\bibitem{Cartis_2019aa}
	{\sc C.~Cartis, J.~Fiala, B.~Marteau, and L.~Roberts}, {\em Improving the
		flexibility and robustness of model-based derivative-free optimization
		solvers}, ACM Trans. Math. Software, 45 (2019),
	\url{https://doi.org/10.1145/3338517}, \url{https://doi.org/10.1145/3338517}.
	
	\bibitem{Cartis_2021aa}
	{\sc C.~Cartis, L.~Roberts, and O.~Sheridan-Methven}, {\em Escaping local
		minima with local derivative-free methods: a numerical investigation},
	Optimization, 0 (2021), pp.~1--31,
	\url{https://doi.org/10.1080/02331934.2021.1883015},
	\url{https://doi.org/10.1080/02331934.2021.1883015}.
	
	\bibitem{Conn_2008part1}
	{\sc A.~R. Conn, K.~Scheinberg, and L.~N. Vicente}, {\em Geometry of
		interpolation sets in derivative free optimization}, Mathematical
	programming, 111 (2008), pp.~141--172.
	
	\bibitem{Conn_2008part2}
	{\sc A.~R. Conn, K.~Scheinberg, and L.~N. Vicente}, {\em Geometry of sample
		sets in derivative-free optimization: polynomial regression and
		underdetermined interpolation}, IMA Journal of Numerical Analysis, 28 (2008),
	pp.~721--748.
	
	\bibitem{Conn_2009book}
	{\sc A.~R. Conn, K.~Scheinberg, and L.~N. Vicente}, {\em Introduction to
		derivative-free optimization}, SIAM, 2009.
	
	\bibitem{Fuhrlander_2020aa}
	{\sc M.~Fuhrländer, N.~Georg, U.~Römer, and S.~Schöps}, {\em Yield
		optimization based on adaptive {Newton}-{Monte} {Carlo} and polynomial
		surrogates}, Int. J. Uncert. Quant., 10 (2020), pp.~351--373,
	\url{https://doi.org/10.1615/Int.J.UncertaintyQuantification.2020033344}.
	
	\bibitem{Fuhrlander_Testset}
	{\sc M.~Fuhrländer and S.~Schöps}, {\em Hermite least squares based on
		{BOBYQA}}, 2022, \url{ttps://github.com/temf/HermiteLSB}.
	
	\bibitem{Graeb_2007aa}
	{\sc H.~E. Graeb}, {\em Analog Design Centering and Sizing}, Springer,
	Dordrecht, 2007.
	
	\bibitem{Hermann_HermiteInterpol}
	{\sc M.~Hermann}, {\em Numerische Mathematik}, Oldenbourg Wissenschaftsverlag,
	2011.
	
	\bibitem{Kennedy_1995}
	{\sc J.~Kennedy and R.~Eberhart}, {\em Particle swarm optimization}, in
	Proceedings of ICNN'95-international conference on neural networks, vol.~4,
	IEEE, 1995, pp.~1942--1948.
	
	\bibitem{Kraft_1988aa}
	{\sc D.~Kraft et~al.}, {\em A software package for sequential quadratic
		programming},  (1988).
	
	\bibitem{Loukrezis_WG}
	{\sc D.~Loukrezis}, {\em Benchmark models for uncertainty quantification},
	2019,
	\url{https://github.com/dlouk/UQ_benchmark_models/tree/master/rectangular_waveguides/debye1.py}.
	
	\bibitem{Menhorn_2022aa}
	{\sc F.~Menhorn, F.~Augustin, H.-J. Bungartz, and Y.~M. Marzouk}, {\em A
		trust-region method for derivative-free nonlinear constrained stochastic
		optimization}, arXiv preprint arXiv:1703.04156,  (2022).
	
	\bibitem{Powell_COBYLA}
	{\sc M.~J. Powell}, {\em A direct search optimization method that models the
		objective and constraint functions by linear interpolation}, Advances in
	Optimization and Numerical Analysis.
	
	\bibitem{Powell_2009aa}
	{\sc M.~J. Powell}, {\em The {BOBYQA} algorithm for bound constrained
		optimization without derivatives}, Cambridge NA Report NA2009/06, University
	of Cambridge, Cambridge,  (2009), pp.~26--46.
	
	\bibitem{Powell_LINCOA}
	{\sc M.~J. Powell}, {\em On fast trust region methods for quadratic models with
		linear constraints}, Mathematical Programming Computation, 7 (2015),
	pp.~237--267.
	
	\bibitem{Sauer_1995aa}
	{\sc T.~Sauer and Y.~Xu}, {\em On multivariate hermite interpolation}, Advances
	in Computational Mathematics, 4 (1995), pp.~207--259.
	
	\bibitem{Ulbrich_2012aa}
	{\sc M.~Ulbrich and S.~Ulbrich}, {\em Nichtlineare Optimierung},
	Birkh\"{a}user, 2012.
	
\end{thebibliography}


\end{document}